\documentclass[12pt,letterpaper]{article}

\usepackage[dvipsnames]{xcolor}
\usepackage[margin=1in,footskip=0.5in]{geometry}
\usepackage[T1]{fontenc}
\usepackage{verbatim}
\usepackage{amsthm}
\usepackage{amsmath}
\usepackage{amssymb}
\usepackage{caption}
\usepackage{natbib}
\usepackage{bbm}
\usepackage{setspace}
\usepackage{subfigure}
\usepackage{xcolor}
\usepackage{graphics}
\usepackage{epsfig}
\usepackage{geometry}
\usepackage{accents}

\usepackage{enumitem}
\usepackage{setspace}
\usepackage{array}
\newcolumntype{x}[1]{>{\centering\arraybackslash\hspace{0pt}}p{#1}}
\usepackage{tikz}
\usetikzlibrary{patterns}
\usetikzlibrary{patterns.meta}
\usepackage{pifont}
\usepackage{avant}

\usepackage[hypertexnames=false]{hyperref}

\newcolumntype{L}[1]{>{\raggedright\let\newline\\\arraybackslash\hspace{0pt}}m{#1}}
\newcolumntype{C}[1]{>{\centering\let\newline\\\arraybackslash\hspace{0pt}}m{#1}}

\makeatletter \renewenvironment{proof}[1][\proofname]
{\par\pushQED{\qed}\normalfont\topsep6\p@\@plus6\p@\relax\trivlist\item[\hskip\labelsep\bfseries#1\@addpunct{.}]\ignorespaces}{\popQED\endtrivlist\@endpefalse} \makeatother

\newtheorem{theorem}{Theorem}

\newtheorem{assumption}{Assumption}

\newtheorem{corollary}{Corollary}

\theoremstyle{definition}

\theoremstyle{plain}

\newtheorem{lemma}{Lemma}

\newtheorem{proposition}{Proposition}
\theoremstyle{remark}

\theoremstyle{plain}

\usepackage{multirow}
\usepackage{chngcntr}
\usepackage{apptools}
\AtAppendix{\counterwithin{lemma}{section}}
\AtAppendix{\counterwithin{proposition}{section}}
\AtAppendix{\counterwithin{theorem}{section}}
\AtAppendix{\counterwithin{corollary}{section}}

\definecolor{uncblue}{RGB}{75, 165, 211}
\definecolor{nberblue}{RGB}{0, 90, 155}
\definecolor{resgreen}{RGB}{34, 84, 67}
\definecolor{pennblue}{RGB}{0, 44, 119}
\definecolor{pennred}{RGB}{152, 30, 50}

\hypersetup{
colorlinks=true,
linkcolor=pennblue,
urlcolor=nberblue,
citecolor=pennblue}
\begin{document}
\newcommand{\reals}{\ensuremath{\mathbb{R}}}
\newcommand{\R}{\ensuremath{\mathbb{R}}}
\newcommand{\Z}{\ensuremath{\mathbb{Z}}}
\newcommand{\E}{\ensuremath{\mathbb{E}}}
\newcommand{\expect}{\ensuremath{\mathbb{E}}}
\newcommand{\dint}{\ensuremath{\,\mathrm{d}}}
\newcommand{\signal}{\ensuremath{s}}
\newcommand{\shirk}{\ensuremath{y}}
\newcommand{\signals}{\ensuremath{\mathbf{s}}}
\newcommand{\shirks}{\ensuremath{\mathbf{y}}}

\newcommand{\capa}{\ensuremath{k}}
\newcommand{\capacost}{\ensuremath{c}}
\newcommand{\outside}{\ensuremath{\mathcal{O}}}

\newcommand{\qmax}{\ensuremath{q_{max}}}
\newcommand{\subsidy}{\ensuremath{s}}

\newcommand{\fee}{\ensuremath{m}}

\newcommand{\type}{\ensuremath{\theta}}

\newcommand{\intfullsupport}{\ensuremath{\int_{\underline{\type}}^{\overline{\type}}}}

\newcommand{\typespace}{\ensuremath{{[\underline{\type}},{\overline{\type}}]}}

\newcommand{\menu}{\ensuremath{\mathcal{K}}}
\newcommand{\kfloor}{\ensuremath{\mathcal{K}^{\text{floor}}}}

\title{Regulating a Monopolist through Capacity Control}

\author{Dihan Zou\thanks{Economics Department, University of North Carolina at Chapel Hill. Email: \href{mailto:dihan@email.unc.edu}{\texttt{dihan@email.unc.edu}}.\newline
I am deeply indebted to Gary Biglaiser, Yunzhi Hu, Peter Norman, Can Urgun, and especially Fei Li, for their unfailing encouragement and guidance. I am grateful to Roberto Corrao,
Paul Heidhues,
Drew van Kuiken,
Qingmin Liu, Jo\~{a}o Thereze, Huseyin Yildirim,
and Nathan Yoder for their helpful comments. All errors are my own.
}}
\date{\today
\\ \href{https://drive.google.com/file/d/1u5heWZuM7NFp0ftcx7xA0FY-2vVUrnyW/view?usp=sharing}{Click here for the latest version}
}

\maketitle

\begin{abstract}
We study monopoly regulation when capacity is contractible but subsequent output is not. A firm privately observes its production cost before installing costly capacity and then choosing output, while the regulator cannot use transfers.
We provide sufficient conditions under which a single capacity floor implements an optimal mechanism with full participation. The floor raises output by changing investment while leaving utilization to the firm after installation costs are sunk. Under additional conditions, the allocation has three regions: laissez-faire, full utilization of the floor, and underutilization. High-cost types knowingly install capacity they will not fully use, yet produce more than under laissez-faire. This physical waste is a necessary evil of optimal regulation that expands service. 
We also show that underutilization can persist in optimal regulation with capacity-contingent subsidies or a common entry fee.
\end{abstract}

\begin{description}
{\small
\item[Keywords:] monopoly regulation, capacity, non-contractible output, delegation, underutilization

\item[JEL Classification Codes:] L43, L5, D82
}

\end{description}

\newpage



\section{Introduction}
Monopoly regulation is challenging when the production cost is unknown. The optimal policy can depend on the form of information asymmetry and
restrictions on policy instruments.
Classical models allow direct regulation of output through price controls, quantity controls, or incentive schemes (see \cite{ArmstrongSappington07HIO}).
In practice, the regulator may instead be able to contract only on installed production capacity, while subsequent output remains outside the scope of the contract. This distinction creates a novel friction in the context of monopoly regulation: the ability to require a facility does not imply the ability to require every unit of service that the facility could provide. How should a privately informed monopolist be regulated when the available instrument governs costly capacity, but the firm's eventual output remains outside the contract?

Several institutional arrangements illustrate this separation. In the Alaska-Hawaiian merger, an agreement governing the U.S. Department of Transportation's approval of transferred international route authorities requires the combined carrier to maintain minimum capacity on specified overlap routes, where it has substantial market power. Domestic airline fares are legally deregulated, and the relevant commitment is expressed in scheduled flights or seats, rather than passengers carried.\footnote{
The commitment lasts for six
years and defines the minimum as at least 90 percent of the number of flights
or 90 percent of the number of seats operated in the year-to-date period
through August 31, 2024, subject to the agreement's exceptions. See
\url{https://www.transportation.gov/sites/dot.gov/files/2025-06/Alaska_Hawaiian_merger_agreement.pdf}.
Domestic airline fares are deregulated; see
\url{https://www.transportation.gov/policy/aviation-policy/airline-rules-fares}.
}
A 2021 amendment to New York City's cable franchise with Verizon required 500,000 additional residential units to reach ``Standard Installation'' status over 2020--2023. This obligation concerned service availability, rather than the number of households that ultimately subscribed.\footnote{
The franchise also contains reporting requirements and remedies. See
\url{https://www.nyc.gov/assets/oti/downloads/pdf/cable-tv-franchises/verizon-cable-franchise-amendment-2021.pdf}.
} 
For a station-front redevelopment near Osaka's Shin-Imamiya Station, the city required the selected developer to build a hotel with at least 400 guest rooms. The relevant project conditions specified room capacity and continued hotel use, but did not condition compliance on room rates or realized occupancy.\footnote{
The conditions also required the guest rooms to average at least 30 square
meters, the hotel to open within six years, and hotel use to continue for ten
years. See
\url{https://www.city.osaka.lg.jp/seisakukikakushitsu/cmsfiles/contents/0000406/406626/280728_keiyakukousyoujian.pdf}
and
\url{https://www.city.osaka.lg.jp/seisakukikakushitsu/page/0000406626.html}.
}
These examples motivate the distinction between capacity and use. In many
regulatory settings, realized output may be observable yet remain outside the
relevant contractual instrument, while capacity is readily specified and
verified.

We study optimal regulation under this instrument restriction in the asymmetric information framework of
\citet{BaronMyerson82}.
A monopolist privately observes its marginal production cost (its ``type''), installs costly capacity, and then chooses output up to that capacity. The regulator values consumer surplus and the firm's profit. In the baseline, transfers are unavailable and regulation consists of a menu of permissible capacities, including the option not to operate. The firm selects its preferred capacity from this menu with full knowledge of its cost and of its subsequent output choice.

Capacity regulation works through a wedge between the firm's investment and
production decisions. When choosing unrestricted (i.e., \textit{laissez-faire}) capacity, the firm accounts
for both production and capacity costs. Once capacity has been installed,
that expenditure is sunk: producing another unit uses available capacity
without incurring another unit of installation cost. A capacity floor can
therefore raise output above laissez-faire by making the firm install more
than it would privately choose. However, this effect has a limit. 
Even when capacity is fully utilized, further expansion can lower welfare for some types if the loss in firm profit outweighs the consumer-surplus gain in the regulator's objective. Once the floor exceeds the firm's preferred continuation output, further investment no longer expands service.
The regulator must therefore weigh the output gains from a common requirement against the costs of both fully utilized expansion and unused
capacity.

Our first result identifies sufficient conditions under which a single capacity floor implements an optimal mechanism  with every type participating
(\autoref{thm_capa_floor}). The regulator thus uses only a simple minimum investment requirement, despite being allowed to offer richer menus.
Types whose preferred capacities exceed the floor remain at laissez-faire;
the floor constrains the remaining types. The sufficient conditions require concave
consumer surplus and prevent the density of high-cost types from declining
too quickly. They order the gains from expanding output toward types that
would otherwise produce less and support intervention through a lower bound.
Under these conditions,
the theorem also shows that the benefits of such intervention can be obtained without excluding any type. 
The optimal capacity floor therefore resolves a trade-off only along the intensive margin: too much capacity for inefficient high-cost types versus too little capacity for efficient low-cost types. Here, excessive capacity can take the form of costly expansion under full
utilization or investment in unused capacity.

Our second result establishes conditions under which an optimal floor is
underutilized by a positive measure of high-cost types
(\autoref{prop_underuse}). An additional condition produces three
positive-length regions. Low-cost types choose their laissez-faire capacity
above the floor. Intermediate-cost types install the floor and fully use it.
High-cost types install the same floor but produce less. The types affected by regulation produce more than under laissez-faire: regulation has relaxed
their investment constraint, even though it cannot push their output beyond
their continuation monopoly output. A low utilization rate therefore need not mean
that the policy has failed to expand service.

The unused capacity is a real cost of implementing the policy. If the
regulator could simply reduce each slack firm's requirement while preserving
everyone else's choices, it would avoid this waste. Private information renders
that operation infeasible. A lower capacity offered to a high-cost type is
also available to other types, and changes their incentives to comply with
the larger requirement. Without transfers, the regulator cannot freely
separate the financial burden of a capacity assignment from the capacity
itself. Under the stated conditions, the best common floor consequently
imposes slack on some types while inducing greater service from others. The
result concerns predictable, optimal underutilization: cost is known when the firm
invests, and no subsequent demand shock is needed to leave capacity idle. Therefore, underutilization is a necessary evil that accompanies optimal regulation that expands service.

The extensions distinguish two forms of fiscal authority. When subsidies
may depend on capacity but taxes are infeasible, we characterize an
optimal allocation that reimburses low-cost firms for capacity
above a common floor (\autoref{prop_subsidy}). These firms install enough to
produce their continuation monopoly outputs. High-cost firms receive no
subsidy and underutilize the floor. Under the proposition's conditions, the
subsidized lower branch has probability equal to the regulator's weight on
firm profit, and the underutilizing types account for the complementary probability. This is a separate optimal mechanism in an enlarged instrument
set: the capacity requirement remains relevant when capacity-contingent subsidies are allowed. Subsidies do not fully replace the capacity floor's screening role.

A common entry fee affects the extensive margin
(\autoref{prop_entry_fee}). The firm learns its cost before deciding whether
to pay for access to the entire capacity menu. For a fixed menu, the fee
does not affect comparisons between capacities but screens participation. An
optimal policy can be represented by an entry cutoff and a floor. Their
joint choice matters because tightening the floor lowers the marginal
entrant's profit and hence the fee that can be collected from every entrant.
Underutilization can persist, but its sufficient condition accounts for this
lost revenue. The comparison distinguishes authority to collect an access
payment from authority to adjust payments across capacity choices.

\paragraph{Related Literature.}
This paper contributes to three strands of literature: monopoly regulation
under asymmetric information, mechanism design with noncontractible actions,
and money burning in delegation.

\smallskip

\noindent{\it Monopoly Regulation under Asymmetric Information.}
Our analysis contributes to the design of simple regulatory policies when
the available instruments are restricted. The canonical screening benchmark
regulates price or output using transfers \citep{BaronMyerson82}.
Simple policy features can arise when transfers remain available.
Using a minimax-regret criterion, \citet{GuoShmaya25RobustRegu} derive
policies that combine average-revenue caps and bounded subsidies.
Taking a prior-free approach, \citet{MishraPatil2025JET} show that
undominated mechanisms have a floor on operating output.
For consumption externalities, \citet{Kang26RobustRegulation} shows that
quantity controls can maximize worst-case welfare even when nonlinear
pricing is available, if the association between preferences and marginal
externalities is unknown.

A complementary line of research examines restrictions on regulatory
instruments, often prescribing simple policies. \citet{BFV17RANDPriceQuantity} compare price and quantity
regulation, with transfers available under either instrument.
\citet{WeiZou26} show that a progressive, soft price cap can be optimal when
taxes are feasible but subsidies are prohibited. Our baseline excludes
transfers and allows the regulator to contract only on capacity.
\citet{AM08RES} characterize optimal delegation and apply their analysis to
monopoly regulation without transfers. Our capacity-menu problem belongs
to this literature because the regulator chooses the firm's permissible
capacities. \citet{AmadorBagwell22Regulation} provide an important benchmark
for contractible-output regulation without transfers, establishing
conditions for optimal price-cap regulation with ex-post participation and
possible exclusion.\footnote{\citet{KolotilinZapechelnyuk25} connect
persuasion and delegation and use this connection to derive further
conditions for optimal price-cap regulation with voluntary participation.}
Their price cap is equivalent to a floor on output
conditional on operation. In our model, a capacity floor leaves the firm free
to choose a smaller continuation output, so pooling in installed capacity
need not imply pooling in service.

Closest to our approach is \citet{Xia25}, who also studies no-transfer
regulation when the regulator controls only part of the firm's decision
problem. His regulator controls allowable qualities and coverage while
leaving the firm to choose its price--quality menu subject to those
restrictions.\footnote{
Related work studies restrictions on downstream contracting. 
In a dynamic matching market, \citet{BhaskarMcClellan26} derive optimal restrictions on bilateral contracts that link permitted termination fees to wages. 
\citet{TherezeVaidya26} characterize the outcomes implementable when a principal restricts the contract menus an informed delegate may offer a privately informed agent.
}
Quality and coverage rules constrain delivered service. Our contractible object is instead costly installed capacity, whose subsequent utilization remains the firm's choice.

\smallskip

\noindent{\it Mechanism Design with Noncontractible Actions.}
Our analysis also relates to mechanism design in which agents'
noncontractible choices interact with private information.
In their seminal contribution, \citet{LaffontTirole86JPE}
study hidden cost-reduction effort with monetary contracts contingent on
output and observed accounting costs. \citet{Krahmer25Flexible} studies
monopoly regulation with covert investment in the distribution of production
costs before the firm privately observes its realized cost.
\citet{CamboniCarnehl25} study certification based on verifiable inputs and
outputs to shape investment incentives through market beliefs.
Our model features ex post moral hazard: the firm knows its production cost
before installing contractible capacity and then chooses noncontractible
output after installation costs are sunk.\footnote{Related work studies
price caps and complementary capacity payments under demand uncertainty
\citep{LemusMoreno17}, optimal network expansion with price and revenue
instruments \citep{WillemsZwart18}, and investment-contingent taxation when
resource extraction is unobservable \citep{Osmundsen02}. In contrast, our
work establishes conditions under which a single capacity floor with unused
capacity arises under optimal regulation. This slack reflects screening
incentives without requiring demand surprises.}
\citet{CorraoFlynnSastry23} analyze nonlinear pricing when buyers can consume
less than they purchase and sellers' payoffs depend on usage. Their
underutilization concerns a contractual allotment: with monotone provision
costs, the seller need not produce unused goods. Our slack instead consists
of installed physical capacity whose resource cost has already been incurred.

\smallskip

\noindent{\it Money Burning in Delegation.}
Methodologically, our analysis builds on the delegation approach initiated
by \citet{Holmstrom77}. \citet{AmadorWerningAngeletos06} and \citet{AB13ECMA}
develop methods for identifying optimal interval restrictions.
\citet{AmbrusEgorov17} and \citet{AmadorBagwell20MoneyBurning} show that costly
nonmonetary actions can be useful incentive instruments in delegation.
\citet{PatelUrgun25} combine verification and money burning in allocation
without transfers. A closely related contribution is
\citet{GonzalezMorinMartimort26}, who study procurement when the principal
privately observes a nonverifiable signal about the agent's cost.
Money burning can complement transfers by supporting credible
signal-dependent payments. In our setting, costly capacity requirements
provide screening incentives when transfers are restricted, linking
productive intervention and resource burning through the same physical
requirement.

\paragraph{Organization.}
Section \ref{sec_model} presents the environment. Section \ref{sec_preliminary} derives preliminary results.
Section \ref{sec:optimal} characterizes the optimal floor and its
underutilization. Section \ref{sec_extension} studies the joint design of capacity control and monetary incentives. Section \ref{sec_conclusion} concludes. All proofs are in the appendix.

\section{Environment}\label{sec_model}

Consider a single-product market with inverse demand $P(q)$, where
$q\in[0,\qmax]$ denotes output.  A monopolist first installs capacity
$\capa\geq0$ and then chooses output
$q\in[0,\min\{\capa,\qmax\}]$, equivalently setting the market price
$P(q)$.  The firm has constant marginal production cost
$\type\in[\underline{\type},\overline{\type}]\subseteq\mathbb R_+$, drawn
from a distribution $F$ with density $f$.  The realization of $\type$ is
privately observed by the firm.  Installing capacity $\capa$ costs
$\capacost(\capa)$. The firm's outside option is zero.

The regulator chooses a menu of permissible capacities but cannot prescribe or
contract on output.  No transfer between the regulator and the firm is allowed
in the baseline model.\footnote{
We relax this assumption in Section~\ref{sec_extension}.
}  
Let $V(q):=\int_0^qP(x)\dint x$ denote gross consumer value. Note that $qP(q)$ is the firm's revenue. So consumer
surplus is $V(q)-qP(q)$.  If a type-$\type$ firm installs capacity $\capa$
and produces output $q$, ex post welfare is
\[
W(q, \capa; \type) := V(q)-qP(q)+\alpha\Pi(q, \capa; \type),
\]
where $\alpha\in[0,1]$ is the welfare weight on firm profit
$\Pi(q, \capa; \type)=qP(q)-\type q-\capacost(\capa)$.  Equivalently, the ex post welfare can be decomposed into a term representing the net social gains from trade and a term accounting for the redistribution motive, i.e., 
\[
W(q, \capa; \type)  =  \underbrace{V(q)-\type q-\capacost(\capa)}_{\text{net social gains from trade}}\quad -\underbrace{(1-\alpha)\Pi(q, \capa; \type)}_{\text{penalty on firm profits}}.
\]

The regulator commits to a closed capacity menu $\menu$ satisfying
$0\in\menu\subseteq\R_+$.  After observing its type $\type$, the firm selects
$\capa_\type\in\menu$ and then chooses
\[
q_\type\in\arg\max_{q\in[0,\min\{\capa_\type,\qmax\}]}
\{qP(q)-\type q\}.
\]
The capacity cost is sunk at the output stage.  Anticipating these choices, the
regulator maximizes $\E_F[W(q_\type,\capa_\type;\type)]$.

The timing separates the decision the regulator can discipline from the decision it ultimately cares about. 
A capacity requirement can rule out small installed capacities, but it cannot compel the firm to use every unit it installs.
Once installation is complete, the
firm compares marginal revenue with production cost alone. Throughout, noncontractibility
refers to this restriction on the regulator’s instruments: output may be observed in practice without being a variable on which the relevant agreement can condition an obligation
or payment. The regulator can enforce the chosen capacity, but cannot supplement that
requirement with an output obligation.

There is no uncertainty after the firm learns its type. In particular, if a firm installs capacity it does not fully use, the firm anticipates that outcome when it invests. Capacity also has
no direct social value in the objective: its benefit comes entirely through production. These
features isolate the incentive role of physical slack from precautionary investment and from
a social preference for reserve capacity. The outside option remains available because $0 \in \menu$.

\medskip

The following assumptions give unique interior operating and investment choices under
laissez faire and permit comparison across cost types.

\begin{assumption}\leavevmode\label{asmp:maintained}
\begin{enumerate}[label=(\roman*)]
\item The inverse demand $P$ is twice continuously differentiable on $(0, \qmax)$, with a finite $\qmax > 0$, and satisfies $P'(q)<0$.
\item The revenue function $qP(q)$ is strictly concave on $[0,\qmax]$,
i.e., $(qP(q))''<0$.  Both $qP(q)$ and
$V(q)=\int_0^qP(x)\dint x$ are finite for every $q>0$ and extend
continuously to zero, with $qP(q)|_{q=0}=V(0)=0$.
\item The capacity cost $\capacost\colon[0,\infty)\to\R_+$ is twice
continuously differentiable, with $\capacost(0)=0$,
$\capacost'(0)\geq0$, $\capacost'(\capa)>0$ for $\capa>0$, and
$\capacost''(\capa)>0$ for $\capa>0$.  Moreover,
\[
\lim_{q\downarrow0}(qP(q))' > 
\overline{\type} + \capacost'(0), \quad 
\lim_{q\uparrow\qmax} (qP(q))'<\underline{\type}.
\]
\item The density $f$ is absolutely continuous and satisfies $f(\type) > 0$ for all $\type \in \typespace$.
\end{enumerate}
\end{assumption}
\autoref{asmp:maintained} has two immediate implications.  First, the firm's
output choice is unique conditional on its capacity choice and its type. Define the \emph{continuation output} by
\[
q^c(\capa,\type) := \arg\max_{q \in [0, \min\{\capa, \qmax\}]}\{qP(q) - \type q\},
\]
where the capacity cost $\capacost(\capa)$ is sunk.  Let
$q^c(\type):=q^c(+\infty,\type)$ denote the unconstrained continuation
output.  Strict concavity of revenue implies this choice uniquely\footnote{
See Appendix \ref{apx:q-m-q-lf} for details.
} satisfies the interior optimal condition
$(qP(q))'=\type$.  Substituting the firm’s output choice gives its market profit and the regulator’s welfare as
\[
\pi(\capa, \type) := \Pi(q^c(\capa, \type), \capa; \type) =  q^c(\capa, \type) P(q^c(\capa, \type)) - \type q^c(\capa, \type)- \capacost(\capa),
\]
and
\[
w(\capa, \type):= W(q^c(\capa, \type), \capa; \type) = V(q^c(\capa, \type)) - \type q^c(\capa, \type) - \capacost(\capa) - (1-\alpha) \pi(\capa, \type).
\]

Second, \autoref{asmp:maintained} (iii) ensures that every type participates under laissez-faire, when there is no restriction on 
permissible capacities ($\menu=\R_+$).  This makes the entire type space relevant
for policy design. Indeed, when the only instrument available is capacity control,  types that do not participate even without regulation will not participate under any regulatory policy.\footnote{
That is, if some types were inactive even under laissez-faire, one could instead truncate
the type space without changing the policy for active types.
}
Under the laissez-faire policy $\menu = \R_+$, 
each type chooses tight capacity $\capa=q$ and therefore internalizes the
capacity cost.  Define its \emph{laissez-faire output} by
\[
q^{LF}(\type):= \arg\max_{q \in [0, \qmax]} \{qP(q) - \type q - \capacost(q)\}.
\]
Under \autoref{asmp:maintained}, $q^{LF}(\type)$ uniquely satisfies
$(qP(q))'-\capacost'(q)=\type$.  The continuation choice
excludes the already sunk capacity cost, so
$q^{LF}(\type)<q^c(\type)$.  Capacity control can therefore raise output by
making the firm's installed capacity exceed its laissez-faire choice.

\paragraph{Regulator's Problem.} By the revelation principle, we can restrict
attention to direct mechanisms $(\capa(\type),\Pi(\type))$, where
$\capa(\type)$ is assigned following a truthful report and $\Pi(\type)$ is
the corresponding profit.\footnote{
We restrict attention to admissible mechanisms whose capacity schedule
\(\capa(\cdot)\) is continuously differentiable on the intervals of a finite
partition of \(\typespace\), with jumps permitted at partition points.
}
A type $\type$ that reports $\hat{\type}$ earns
$\pi(\capa(\hat{\type}),\type)$.  Because transfers are unavailable, truthful
profit must satisfy
\begin{equation}
    \Pi(\type) = \pi(\capa(\type), \type), \label{N}\tag{N}
\end{equation}
so the firm's payoff equals the producer surplus generated in the market.  A
direct mechanism is feasible if, for all $\type,\hat{\type}$, the firm with type $\theta$ has no incentive to misreport as type $\hat{\theta}$, i.e., \[
\Pi(\type) \geq \pi(\capa(\hat{\type}), \type), \label{IC}\tag{IC}
\]
and the firm makes nonnegative profit when truthfully reporting, i.e., 
\[
\Pi(\type) \geq 0. \label{IR}\tag{IR}
\]
The regulator's problem is therefore
\begin{equation}
      \max_{\capa(\cdot), \Pi(\cdot)} \intfullsupport w(\capa(\type); \type) \dint F(\type) \label{direct}\tag{P}
 \end{equation}
subject to \eqref{IC}, \eqref{IR}, and \eqref{N}.

Appendix \ref{apx} gives an equivalent optimal-control formulation of \eqref{direct}. 

\section{Preliminary Analysis} \label{sec_preliminary}
 A mechanism is
\emph{feasible} if it satisfies \eqref{IC}, \eqref{IR}, and \eqref{N}.  These constraints order not only production but also the physical capacity needed to support it.

\begin{lemma}\label{lem:monotone-welfare}
In any feasible mechanism, capacity $\capa(\type)$, induced output
$q^c(\capa(\type),\type)$, truthful profit $\Pi(\type)$, and ex post
welfare $w(\capa(\type),\type)$ are weakly decreasing in type.
\end{lemma}
The no-transfer restriction extends the usual output ordering to installed
capacity. Pairwise incentive comparisons imply that a higher-cost type cannot
produce more than a lower-cost type. Suppose, nevertheless, that it installed
more capacity. Its output would then fit within the lower-cost type's smaller
capacity.\footnote{
Pairwise incentive
constraints (which require nonincreasing output in type) and the firm's continuation output choice (nondecreasing in capacity) 
imply that the two capacities must
induce the same output throughout the intervening type interval. 
}
By choosing that capacity and producing the same output, it would
save installation costs without sacrificing operating profit. With no transfers
to offset this saving, the proposed assignment violates incentive compatibility.

The same argument rules out assigning different capacities to types that
produce the same output. The feasible separation is indeed limited.
On a smoothly separating part of a feasible allocation, each type must choose its privately optimal capacity. Distorting capacity away from laissez faire therefore requires pooling over type intervals, although output can still vary within a capacity pool. Appendix \ref{apx} establishes this restriction formally.

Profit monotonicity also follows directly from incentive compatibility.
A lower-cost type can choose the capacity assigned to a higher-cost type
and reproduce its output at a weakly lower production cost. Its truthful
profit must therefore be at least as high. Consumer surplus is increasing
in output, so both consumer surplus and the firm's profit weakly decline
with type. Since the regulator places nonnegative weight on both,
welfare is also weakly decreasing in type.

\medskip

We call a capacity menu $\menu$ \emph{optimal} if it implements an optimal
mechanism for Problem \eqref{direct}.  The following lemma establishes that some intervention is always desirable.  Without intervention,
the firm chooses $q^{LF}(\type)$ and tight capacity
$\capa(\type)=q^{LF}(\type)$.
\begin{lemma}\label{lem:full-delegation}
    Laissez faire ($\menu=\R_+$) is not an optimal policy.
\end{lemma}
The reason is the wedge between the regulator's objective and the firm's
objective.  Without intervention, each type chooses $q^{LF}(\type)$ by equating marginal revenue with marginal production and capacity cost, which leads to standard monopoly distortion because  the marginal social benefit of increasing output from $q^{LF}(\type)$ is strictly above the marginal revenue.
A sufficiently small increase in tight capacity from the laissez-faire choice therefore raises welfare.
To be specific, consider the capacity-floor menu
$\menu^{\mathrm{floor}}(K):=\{0\}\cup[K,\infty)$ with
$K>q^{LF}(\overline{\type})=: \hat q$,\footnote{
No regulation is a trivial capacity floor policy with $K \leq q^{LF}(\overline{\type})$.
} which rules out operation at capacities below $K$.  If $K$ is only slightly
above $\hat q$, it affects a small interval of high-cost types, which choose
tight capacity and output $K$.  Their loss in profit is second order at the
laissez-faire choice, whereas consumers receive a first-order gain from the
increase in output.  Thus a local floor perturbation improves welfare even
without transfers.

The capacity floor illustration sheds light on the idea behind the optimal policy. Because $q^{LF}$ is strictly decreasing, a nontrivial floor
intersects the laissez-faire schedule at most once. Across $\typespace$, if
$K\leq q^{LF}(\type)$, type $\type$ installs and produces
$q^{LF}(\type)$ and is unaffected.
Conditional on participation, a type with \(K>q^{LF}(\type)\) installs \(K\) and produces \(q^c(K,\type)>q^{LF}(\type)\). For the small floor perturbation considered above, every type participates and the resulting output gains yield an ex ante welfare improvement.

\section{Optimal Capacity Control}
\label{sec:optimal}
This section characterizes the optimal mechanism.  \autoref{thm_capa_floor}
gives sufficient conditions under which a nontrivial capacity floor
$\{0\}\cup[K, \infty)$ is optimal.  \autoref{prop_underuse} then provides additional
conditions under which a positive measure of types underutilize that floor.

\subsection{Why a single floor is optimal}

\begin{theorem}\label{thm_capa_floor}
    There exists a capacity floor $\{0\}\cup[K^*, \infty)$ with $K^* \in (q^{LF}(\overline{\type}), q^c(\underline{\type})]$ that implements an optimal mechanism if
    \begin{equation}
       \frac{-f'}{f} \leq \frac{\alpha}{-\hat{q}P'(\hat{q})} \text{ a.e. on } \typespace,\; \text{ and } P'(q) + qP''(q) \geq 0 \text{ on }(0, \qmax).\label{FL} \tag{FL}
    \end{equation}
Moreover, $\pi(K^*, \overline{\type})\geq 0$, so every type participates.
\end{theorem}

The optimal policy is a simple capacity floor even though the regulator
can choose a much richer set of capacity restrictions. In particular,
it could exclude intermediate capacities or offer several separated
permissible ranges. Under \eqref{FL}, these alternatives do not improve
on a single floor within the admissible class. The regulator therefore
removes small positive capacities while preserving the firm's discretion
above the minimum.

The bounds on $K^*$ describe the scope of intervention. The lower bound
$K^*>q^{LF}(\overline{\type})$ makes the policy nontrivial: the floor binds
for a high-cost segment. Types with $q^{LF}(\type)<K^*$ install the floor
and produce $q^c(K^*,\type)>q^{LF}(\type)$, while types whose preferred
capacities satisfy $q^{LF}(\type)\geq K^*$ retain their laissez-faire
capacity and output. The upper bound $K^*\leq q^c(\underline{\type})$
rules out a floor that exceeds every type's continuation output.

Condition \eqref{FL} combines restrictions on demand and the cost density.
The demand restriction orders the marginal consumer-surplus gains from
intervention across types. Let  $S(q):=V(q)-qP(q)$ denote consumer
surplus. Since
\[
P'(q)+qP''(q)=-S''(q),
\]
the inequality $P'+qP''\geq0$ is equivalent to concavity of consumer surplus.
Together with $P'<0$, it also implies $P''>0$, so inverse demand must be
convex.\footnote{E.g., constant-elasticity demand functions satisfy this requirement.}
Because $q^{LF}(\type)$ falls with cost, concavity makes
$S'(q^{LF}(\type))$ weakly increasing in $\type$. At laissez faire, this
is also the \textit{ex post} marginal welfare gain from expanding capacity, since the firm's
profit is locally maximized. The immediate benefit from intervention is
therefore ordered toward the high-cost types that produce the least.

The density restriction complements this ordering. The inequality
$-f'/f\leq\alpha/[-\hat qP'(\hat q)]$ holds automatically wherever $f$ is
nondecreasing. Where $f$ is strictly decreasing, $-f'/f$ is its proportional
rate of decline, and the inequality prevents the likelihood of encountering
less efficient firms from falling too sharply. Together, the demand and
density restrictions provide the \emph{ex ante} ordering needed for a single
capacity floor to be optimal, taking account of both the distribution of
types and the regulator's valuation of firm profits.

The density requirement is easier to satisfy when the regulator places
greater weight on the firm's profit (larger $\alpha$) or when marginal
consumer surplus is smaller at the highest-cost type's laissez-faire output
(smaller $-\hat qP'(\hat q)$). Concavity makes this latter quantity an upper
bound on marginal consumer surplus along the entire laissez-faire schedule.
A smaller value therefore lowers the bound on the marginal consumer gains
from expanding output, while a larger $\alpha$ places greater weight on the
profit losses caused by intervention. In this sense, when potential consumer
gains are smaller relative to the regulator's concern for firm profits, the
condition permits a faster decline in the density while still ensuring
that no richer intervention within the admissible class improves on a single
capacity floor.

\paragraph{The anatomy of a capacity floor.}
To clarify the forces behind the characterization, fix a nontrivial floor
$K>q^{LF}(\overline{\type})$. Conditional on participation, a type's
capacity and output take one of three forms:
\[
(\capa(\type),q(\type))
=
\begin{cases}
\bigl(q^{LF}(\type),q^{LF}(\type)\bigr),
& q^{LF}(\type)\geq K,\\[3pt]
(K,K),
& q^{LF}(\type)<K\leq q^c(\type),\\[3pt]
\bigl(K,q^c(\type)\bigr),
& q^c(\type)<K.
\end{cases}
\]
In the first case, the firm's preferred capacity remains available, so the
floor has no effect. In the second, the floor exceeds the firm's
laissez-faire capacity but does not exceed its unconstrained continuation
output. The firm therefore installs $K$ and uses it fully. As long as the
floor remains below $q^c(\type)$, a marginal increase raises capacity and
output one for one. In the third case, the floor exceeds the firm's
unconstrained continuation output. The firm still installs the minimum
permissible capacity but, because installation costs are sunk when production
is chosen, produces only $q^c(\type)<K$.

Because both $q^{LF}(\type)$ and $q^c(\type)$ decrease with type and
$q^{LF}(\type)<q^c(\type)$, these responses follow a natural order. As
production cost rises, the allocation can move from laissez-faire to full
utilization of the floor and then to underutilization, although one or more
regions may be empty for a given $K$. Underutilization does not mean that
regulation has no effect on output: a participating type with slack capacity
still produces $q^c(\type)>q^{LF}(\type)$. Rather, once the floor exceeds
$q^c(\type)$, further increases no longer raise that type's output.

Raising the floor entails two distinct intensive-margin costs: it can
push fully utilized capacity beyond the level justified by the output gains,
and it can require additional capacity that remains unused. Both costs arise
conditional on continued operation.

The first cost can arise even when the firm uses all of its capacity.
While a type remains in the interior of the tight-floor region,
\[
\frac{d}{dK}W(K,K;\type)
=
\underbrace{-KP'(K)}_{\text{marginal consumer-surplus gain}}
+
\alpha
\underbrace{\left[
(KP(K))'-\type-\capacost'(K)
\right]}_{\text{marginal firm profit}}.
\]
At $K=q^{LF}(\type)$, the marginal-profit term is zero by the firm's
laissez-faire first-order condition. A small increase in tight capacity
therefore raises welfare because $-KP'(K)>0$, providing the rationale for
\autoref{lem:full-delegation}. Once $K>q^{LF}(\type)$, however, the floor
pushes the firm beyond its privately optimal capacity, so further expansion
reduces its profit. This loss, weighted by $\alpha$, can exceed the
consumer-surplus gain even though the additional capacity is fully utilized.
At a common floor, the marginal consumer-surplus gain is the same for all
types in the tight region, while the marginal profit loss is larger for
higher-cost types. Therefore, at the optimal floor, a further increase can
reduce welfare for some higher-cost types that fully utilize capacity while
raising welfare for more efficient types in the same region.

The second cost arises when the firm leaves capacity unused. For a type
already underutilizing the floor, output remains at $q^c(\type)$. While the
type remains in this region,
\[
\frac{d}{dK}W(q^c(\type),K;\type)
=
-\alpha\capacost'(K)
\leq0,
\]
with strict inequality whenever $\alpha>0$. Additional capacity now raises
installation costs without increasing output, so there is no marginal
consumer-surplus gain to offset the profit loss.
Therefore, underutilization can occur in the optimal mechanism
only when the gains from raising output for some types outweigh the losses from creating slack capacity for more inefficient types.

To sum up, the downside of raising the
floor takes two forms: additional production whose consumer benefits
may be outweighed by its cost to the firm, and additional investment that
produces no further service. The optimal floor balances these effects across
types, taking account of their probabilities and the regulator's weight on
firm profits. 

An excessively high floor can also impose an extensive-margin cost by
inducing exit. For a type whose laissez-faire capacity lies below $K$, the
best positive permissible capacity is $K$. If $\pi(K,\type)<0$, the type
exits because $0$ also belongs to the menu. Under \eqref{FL}, however,
\autoref{thm_capa_floor} shows that the regulator can attain the optimum
without excluding any type. The remaining trade-off at the selected optimum
therefore concerns the intensive margin, including both costly expansion of
fully utilized capacity and investment in unused capacity.

\paragraph{Why a single floor?}
The preceding trade-off suggests a possible role for more selective
regulation. Could the regulator encourage greater investment from types
whose additional output is valuable while allowing types for whom the
requirement is costly to choose smaller capacities? A menu with several
separated permissible ranges might appear to achieve this. However,
every permissible capacity is available to every type, so the regulator
must account for how firms select among these options. Excluding
intermediate capacities can encourage some firms to invest and produce
more, but can also lead others to choose smaller capacities and produce
less. Whether such a policy improves welfare depends on the resulting
output changes, profit losses, and the distribution of types. \autoref{thm_capa_floor} establishes that, under \eqref{FL}, none improves on a single floor.
Roughly speaking, condition \eqref{FL} favors intervention toward the high-cost end of the type distribution. Evaluated at their respective laissez-faire outputs, less efficient firms generate a weakly larger marginal consumer-surplus gain from additional output and face a lower marginal installation cost because they operate at a smaller scale. The density restriction complements this ordering by preventing the likelihood of encountering less efficient firms from falling too sharply. These comparisons help explain the form of the optimal policy: a single floor raises investment by firms whose preferred capacities are small while preserving discretion above the minimum.

A further possibility is to raise the floor enough to induce shutdown. Such a policy can increase output for some operating types while excluding others.
Under \eqref{FL}, however, any
admissible menu that induces shutdown can be replaced by a floor that
retains all types, preserves the highest-cost type's zero profit, and
yields weakly higher expected welfare. Full participation can therefore
be maintained at an optimum. Outside \eqref{FL}, more selective restrictions or a floor that induces
shutdown may become attractive.

\paragraph{Comparative statics.}
The optimal capacity floor responds systematically to the regulator's
preferences and the cost of investment.\footnote{
Each comparison holds the
other primitives fixed and maintains condition \eqref{FL}. 
If the optimum is not unique, 
the smallest optimal full-participation floor is selected.
} Increasing $\alpha$, the weight on firm profit, weakly lowers the
floor: a higher requirement reduces expected profit, so giving that loss
more weight weakens the case for additional investment and output. Likewise, replacing
$\capacost(k)$ with $\lambda\capacost(k)$ makes the optimal floor weakly
decrease with $\lambda$. Costlier capacity increases the burden of a
requirement and tightens participation. On a fixed cost support, a shift toward higher marginal costs in hazard-rate
order also weakly lowers the floor.\footnote{
The hazard-rate comparison is
$f_2(\theta)/[1-F_2(\theta)]\leq f_1(\theta)/[1-F_1(\theta)]$
for $\theta<\bar\theta$. It holds if $f_2/f_1$ is nondecreasing.
}
At a given floor, conditional on being affected, this shift reduces the expected
output gain from additional capacity and raises the expected profit loss.
Demand growth works in the opposite direction. A parallel upward shift
$P_A(q)=P_0(q)+A$, or a proportional expansion of market size
$P_z(q)=P_0(q/z)$, weakly raises the optimal floor as $A$ or $z$ increases.
Both changes support higher production and allow greater investment while
preserving participation. Appendix~\ref{ext:comparative-statics} proves these
comparisons.

\subsection{When the optimal floor is underutilized}

\autoref{thm_capa_floor} establishes a floor and full participation, but does not require unused capacity.
The next result separates two questions: whether the floor is high enough to leave some
capacity idle, and whether it is low enough to preserve an unaffected segment of efficient
firms. 
 
\begin{proposition}\label{prop_underuse}
    Suppose condition \eqref{FL} holds, and let $\overline{q}:=q^c(\overline{\type})$. If
    \begin{equation}
-\overline{q} P'(\overline{q}) >
\max\left\{
\frac{\capacost(\overline{q})}{\overline{q}},
\alpha\capacost'(\overline{q})
\right\},
\label{U}\tag{U}
\end{equation}
    then there is an optimal capacity floor policy $\{0\}\cup[K^*, \infty)$ that induces a positive measure of types to install capacity $K^*$ and choose output $q(\type) < K^*$ (i.e., underutilization). 
    If, in addition,
    \begin{equation}
         -\overline{q} P'(\overline{q}) < \frac{ \capacost(q^{LF}(\underline{\type}))}{\overline{q}}, \label{LF}\tag{LF}
    \end{equation}
    then there exist cutoffs $\underline{\type}< \type_{FL} < \type_{U} < \overline{\type}$ such that the type space is partitioned into three positive-length regions: the \textit{laissez-faire region} $[\underline{\type}, \type_{FL}]$, the \textit{floor region} $(\type_{FL}, \type_{U}]$, and the \textit{underutilization region} $(\type_U, \overline{\type}]$. Output and capacity are given by
    \[
q(\type) = \begin{cases}
    q^{LF}(\type) & \text{on } [\underline{\type}, \type_{FL}]\\
    K^* & \text{on } (\type_{FL}, \type_U]\\
     q^c(\type) & \text{on } (\type_U, \overline{\type}]
\end{cases}, \quad 
\capa(\type) = \begin{cases}
    q^{LF}(\type) & \text{on } [\underline{\type}, \type_{FL}]\\
    K^* & \text{on } (\type_{FL}, \overline{\type}]
\end{cases}.
\]
\end{proposition}

Condition \eqref{U} combines financial room for expansion with a sufficiently large consumer benefit from additional output. It ensures that the highest-cost type can remain active as the floor crosses its continuation output, and that crossing this threshold raises expected welfare.
For example, a market may contain a substantial group of consumers willing to pay high prices. Even when high production costs limit the firm’s output, serving these consumers can sustain a sizable operating surplus. If installation costs are moderate, this surplus gives the firm room to absorb a binding capacity requirement without exiting. When the consumer benefits from expanding service are sufficiently large, the regulator may find it optimal to raise the floor beyond the highest-cost type’s continuation output, even though that type will leave some installed capacity unused.

 The first inequality in \eqref{U}, $-\overline qP'(\overline q)>
\capacost(\overline q)/\overline q$, says that the highest-cost type can
cover the cost of installing its continuation output. Its operating margin at
that output is $P(\overline q)-\overline\type=-\overline qP'(\overline q)$,
so this condition is equivalent to $\pi(\overline q,\overline\type)>0$.
There is therefore room to increase the floor beyond $\overline q$ without
immediately losing the highest-cost type.

The second inequality, $S'(\overline q)>\alpha\capacost'(\overline q)$,
makes further intervention worthwhile at that threshold. When the floor first
passes $\overline q$, a positive mass of affected types still uses it fully.
Their additional output benefits consumers, while capacity cost enters the
regulator's objective through profit.  The inequality thus ensures that it is profitable to cross the point at which the highest-cost type becomes unconstrained, creating a positive mass
of firms with unused capacity.

Condition \eqref{LF} has a separate role. It places the optimal floor below
the lowest-cost type's preferred capacity, leaving a positive mass at
laissez faire. 
The highest-cost type cannot finance the lowest-cost type’s laissez-faire capacity from its operating profit (e.g., when the types are diverse). A floor preserving full participation must therefore leave some efficient types unaffected.
Together, conditions \eqref{LF} and \eqref{U} ensure that the optimal capacity floor induces a three-region allocation.

\paragraph{Slack Capacity as a Necessary Evil.}
Direct output regulation would translate a required unit into service. A capacity requirement does so only while the production constraint binds.
For a known slack type, reducing installed capacity toward its output saves resources without reducing service. The difficulty is that the regulator does not know that type. 
Adding a lower capacity to the menu makes it available to every firm, including types whose output the original floor was intended to raise.  
The regulator cannot evaluate the saving for the slack type while holding other types' choices fixed by assumption.

Unused capacity constitutes a form of money burning: at a realized allocation
$\capa>q$, the firm incurs the additional resource cost
$\capacost(\capa)-\capacost(q)>0$ without providing additional service.
 Because the firm knows its production cost before investing, it anticipates this unused investment.
The regulator leverages this burden through capacity rather than choosing it
as a separate instrument. For a given floor $K$, the least costly permissible
capacity supporting any positive output $q$ is $\max\{K,q\}$. Its cost can
therefore be written as
\[
\capacost(\max\{K,q\})
=\capacost(q)+\max\{\capacost(K)-\capacost(q),0\}.
\]
The second term is a nonnegative money-burning schedule generated by the
floor: it is positive below $K$ and zero at or above $K$.  
Types need not select different capacities to incur
different amounts of waste: relatively efficient firms can fully use the same floor that less efficient firms underutilize. The incentive role of
unused investment thus arises from the capacity menu and the firm's
subsequent production choice. Under the conditions of
\autoref{prop_underuse}, an optimal floor includes this resource loss while
expanding service: for some firms it compels additional output; 
for others it imposes investment beyond the point where output responds. An additional unit of already idle capacity is wasteful. The menu containing that capacity can nevertheless dominate menus that avoid the waste but weaken output incentives
elsewhere.

\medskip

\begin{figure}[t]
    \centering
   \begin{tikzpicture}[scale=9]
        \footnotesize
    
            \draw [->, thick,gray] (0,0) -- (0,1);
            \draw [->, thick,gray] (0,0) -- (1.2,0);
           
            \node at (-0.08, 1.05) {$q, \capa$};
            \node at (1.22, -0.05) {$\type$};

            \draw[dotted, pennred, thick] (0,0.4) [out = -30, in = 175] to (1, 0.08);

            \node at (1.05, 0.04) {${\color{pennred}q^{LF}}$};

            \draw[dashed, black!50] (0,0.6) [out = -30, in = 175] to (1, 0.2);
            \node at (1.05, 0.2) {${\color{black!50}q^{FB}}$};

\draw[dashed] (0,0.9) [out=-60, in=170] to (1, 0.1);
\node at (0.15, 0.8) {${q^{c}(\cdot)}$};

\draw[line width = 3 pt, nberblue, opacity = 0.4](0.575, 0.15) -- (1, 0.15);

\node at (-0.05, 0.15) {${\color{nberblue!70}K^*}$};

\draw[line width = 3 pt, nberblue, opacity = 0.4](0, 0.4)[out = -30, in = 165] to (0.575, 0.15) -- (1, 0.15);
\node at (0.1, 0.4) {${\color{nberblue}\capa(\cdot)}$};
\draw[line width = 3 pt, nberblue, dashdotted, opacity = 0.4](0, 0.15) -- (0.575, 0.15);

\fill[pennred, opacity=0.3]
(0.795, 0.15) [out = -15, in = 170] to (1, 0.1) -- (1, 0.15);

\draw[very thick](0, 0.4)[out = -30, in = 165] to (0.575, 0.15) -- (0.795, 0.15) [out = -15, in = 170] to (1, 0.1);

\node at (1.05, 0.11) {$q(\cdot)$};
\node at (0.575, -0.05) {$\type_{FL}$};
\node at (0.795, -0.05) {$\type_{U}$};

\node at (0.35, -0.1) {laissez-faire};
\node at (0.67, -0.1) {floor};
\node at (0.98, -0.1) {underutilization};

        \end{tikzpicture}
    \caption{Illustration of an optimal underutilized floor. {\footnotesize The illustration is based on a numerical example with $P(q) = 1.2 - 0.4\log(q)$, $F(\type) = \frac{e^{\type - 1} - 1}{e - 1}$ on $[1,2]$, $\capacost(\capa) = \capa^2$, and $\alpha = 0.6$. The example satisfies \eqref{FL}, \eqref{U}, and \eqref{LF}, so \autoref{thm_capa_floor} and \autoref{prop_underuse} yield the three regions shown. The gray dashed line is the complete-information, tight-capacity benchmark $q^{FB}(\type) := \arg\max_{q\in [0, \qmax]} W(q,q; \type)$. Appendix \ref{apx:q-m-q-lf} provides details of $q^{FB}$.}}
    \label{fig_underuse}
\end{figure}

\autoref{fig_underuse} illustrates how an optimal single capacity floor produces three different responses.  For $\type\leq\type_{FL}$, the firm chooses
$q(\type)=\capa(\type)=q^{LF}(\type)$ weakly above the floor.  For
$\type\in(\type_{FL},\type_U]$, both capacity and output equal $K^*$.
For $\type>\type_U$, capacity remains at $K^*$ but output follows the lower
continuation schedule $q^c(\type)$.  The shaded wedge between capacity and output
is unused capacity.  Although the first-best output schedule $q^{FB}$ lies above the implemented output, a
higher capacity floor cannot close that output gap once the firm's continuation
choice is already unconstrained.

The figure also clarifies the marginal tradeoff behind the optimal floor.
Increasing $K$ raises output for affected types that still use capacity tightly,
but it raises only capacity cost for types that are already underutilizing.
Condition \eqref{U} ensures that the output gain dominates when the floor first
crosses $\overline q$, so some slack is optimal.  As the floor rises further,
the set of slack types expands and the wasteful component becomes more important.
The optimal $K^*$ balances the output gain on the tight segment against the
capacity cost imposed on both tight and slack segments.

\section{Extensions}
\label{sec_extension}
The baseline isolates capacity as the regulator’s only instrument.
In practice,
capacity can also attract financial support: a 2025 Hunan policy provides a construction subsidy to qualifying smart-computing facilities that meet capacity and energy-efficiency thresholds; meanwhile, underutilization of smart-computing capacity has been reported across China.\footnote{
For new or converted
projects with at least 50 PFlops of computing capacity and power usage
effectiveness no greater than 1.3, the cited provision subsidizes 10 percent
of eligible equipment and software construction expenditure, up to RMB
20 million. See
\url{https://www.hunan.gov.cn/topic/yhyshj/zcfb/sjzc/202502/t20250220_33628757.html}. For evidence of underutilization, see \url{https://www.digitalchina.gov.cn/2025/xwzx/qwfb/202507/t20250717_5048811.htm}.
}
This motivates asking how optimal capacity requirements change when monetary instruments become available.

 We now consider two ways of adding a monetary instrument: 
(i) paying subsidies for capacity and 
(ii) charging a common fee for entry. 
Under sufficient conditions, an optimal mechanism can combine an underutilized floor with either instrument. 
A capacity subsidy changes the firm’s investment incentives; a common fee changes its willingness to operate.

\subsection{Capacity control with subsidies}
Suppose the regulator can offer a nonnegative capacity-contingent subsidy $\subsidy(\capa)\geq0$ together with a menu $\menu$.  A type-$\type$ firm
chooses $\capa\in\menu$, receives $\subsidy(\capa)$, and then produces
$q^c(\capa,\type)$. 
The subsidy is sunk at the output stage, so it changes the cost of installing capacity without changing the continuation choice at a given capacity. 
The firm's payoff is $\pi(\capa,\type)+\subsidy(\capa)$.  With subsidies financed one-for-one by the regulator, the ex post welfare is
\[
w^S(\capa, \type) = S(q^c(\capa, \type)) + \alpha \pi(\capa, \type) - (1-\alpha) \subsidy(\capa).
\]
The expenditure buys an incentive to install capacity. Its fiscal cost is partly offset by the value the regulator places on the firm’s receipt, which explains the coefficient $1-\alpha$. 
The physical investment cost remains a resource cost regardless of who finances it. 
Paying a subsidy therefore does not make otherwise unused capacity socially free. 
The regulator solves
\begin{equation}
    \max_{\subsidy(\cdot) \geq 0
    ,\; \menu} \E_F[w^S(\capa, \type)] \label{subsidyDesign}\tag{S}
\end{equation}
subject to
\begin{align}
    &\capa \in \arg\max_{x\in \menu} \{\pi(x, \type) + \subsidy(x)\}. \label{IC-S}\tag{IC-S}
\end{align}
Zero capacity remains available, which ensures individual rationality. 
Subsidies offer a direct way to encourage investment above the laissez-faire level. 
They do not, however, provide the regulator with a way to collect the firm’s information rent. 
The following result shows how this asymmetry can support subsidies for efficient types alongside an underutilized floor for less efficient types.

\begin{proposition}\label{prop_subsidy} Suppose $\alpha \in (0,1)$ and \eqref{FL} holds. Let $\type_\alpha$ be the unique $\alpha$-quantile defined by $F(\type_\alpha) = \alpha$, and define $K_\alpha := q^c(\type_\alpha)$ and $\overline q:=q^c(\overline\type)$. Furthermore, if
\begin{align}
\overline{q}P(\overline{q})-\overline{\type} \cdot \overline{q}
&\geq \capacost(K_\alpha),
\label{R3IR}\\
-q^c(\type)P'(q^c(\type))-\capacost'(q^c(\type))
&\geq
(1-\alpha)\frac{F(\type)}{f(\type)}
\qquad
\text{for }\type\leq\type_\alpha,
\label{R3L}
\end{align}
then there exists an optimal mechanism whose induced output,
capacity, and subsidy are
\begin{align*}
q^*(\type)
=q^c(\type),\;
\capa^*(\type)
=\max\{q^c(\type),K_\alpha\},\;
\subsidy^*(\type) =
\left[
\capacost(q^c(\type))-\capacost(K_\alpha)
\right]_+.
\end{align*}
Types below $\type_\alpha$ receive positive subsidies and install tight capacity; types above $\type_\alpha$ receive zero subsidies and strictly underutilize the capacity floor $K_\alpha$.
\end{proposition}

In the optimal mechanism constructed in \autoref{prop_subsidy}, the
implementing policy combines a capacity floor $K_\alpha$ with reimbursement
of capacity cost above it:
$\subsidy(\capa)=\capacost(\capa)-\capacost(K_\alpha)$ for
$K_\alpha\leq\capa\leq q^c(\underline\type)$, with zero subsidy
otherwise. Every reimbursed capacity therefore leaves the firm paying the
same net installation cost, $\capacost(K_\alpha)$. The firm maximizes its
remaining operating profit by attaining $q^c(\type)$. Conditional on
participation, its smallest optimal capacity is precisely
$\max\{q^c(\type),K_\alpha\}$.

The conditions in \autoref{prop_subsidy} identify a deliberately strong
regime. Condition \eqref{R3IR} ensures that even the highest-cost type can
cover the common net installation cost from operating profit. Condition
\eqref{R3L} makes the gain from expanding efficient firms' output large
enough, relative to capacity cost and the information-rent term, to bring
them all the way to the continuation ceiling. Together with \eqref{FL},
\eqref{R3L} also supports the zero-subsidy branch for higher-cost types.
The quantile $F(\type_\alpha)=\alpha$ determines the switch: in the
designated allocation, a fraction $\alpha$ receives positive subsidies and
a fraction $1-\alpha$ underutilizes the floor. The construction thus shows
how subsidies and capacity control can complement each other in an optimal
mechanism.

For efficient firms, reimbursement finances the investment needed to reach
their large continuation outputs. Once incremental capacity cost is
reimbursed, the firm can attain the output it would choose if installation
costs were already sunk. It still equates marginal revenue with marginal
production cost. Thus $q^c(\type)$ is the continuation ceiling imposed by
noncontractible output and need not equal the regulator's
complete-information benchmark $q^{FB}(\type)$. Subsidizing unused capacity
cannot induce output above this ceiling.

Less efficient firms already reach their preferred continuation output at
the floor, so the designated allocation pays them no subsidy. The no-tax
restriction explains why their physical slack need not disappear. Consider
a high-cost type assigned zero subsidy and capacity above its output.
Reducing capacity to the amount used would save resources without changing
production, but it would also raise the firm's profit. This rent increase
can change the incentives of more efficient types to choose the options intended for less efficient firms.
One way that preserves the promised payoff and associated incentives would be to collect the capacity-cost saving as a tax, which is infeasible
in this setting. Alternatively, unused capacity
absorbs the rent and keeps the exact incentive. Allowing
subsidies alone therefore does not justify simply reducing every slack
allocation to the capacity actually used while keeping the incentive scheme
intact.

\subsection{Capacity control with an entry fee}
Next suppose the regulator can charge a common fee for access to the entire capacity menu.
It announces $(\menu,\fee)$ with $\fee\geq0$.   The firm learns its type before accepting or rejecting the offer. 
Upon acceptance, it pays the nonrefundable fee and then chooses capacity and output.
Rejection yields zero, and fee revenue enters the regulator’s objective one-for-one. 
The fee cannot depend on the report, installed capacity, or realized output. The rationale here is that the fee collection in practice may require a separate authority from that of capacity control. 
Conditional on entry, the fee cancels from every comparison between capacities. It therefore preserves the firm’s investment and production choices for a given menu. On the other hand, entry is screened by the fee. 
Higher operating costs reduce the value of access, so a common fee selects a lower
interval of types. 
The marginal entrant’s gross profit determines the largest fee consistent with that cutoff. The firm accepts this fee after learning its cost.\footnote{
 The pre-type fee studied in Appendix \ref{apx:pretype_fee} is already sunk when operating information arrives, so it leaves subsequent participation
comparisons unchanged.
}

\begin{proposition}
\label{prop_entry_fee}
Suppose condition \eqref{FL} holds.  Then there exists an optimal mechanism implemented by a capacity floor with a common entry fee.

In particular, there exist an entry cutoff $\type_E^*\in(\underline{\type},\overline{\type}]$
and a floor $K^*\geq q^{LF}(\type_E^*)$ such that
\[
    \menu^*=\{0\}\cup[K^*, \infty),\;
    \fee^*=\pi(K^*,\type_E^*)\geq0.
\]
Types $\type\leq\type_E^*$ enter, and types
$\type>\type_E^*$ reject.

\end{proposition}

To see how the instruments interact, define
$\capa_K(\type):=\max\{K,q^{LF}(\type)\}$ and gross profit
$v_K(\type):=\pi(\capa_K(\type),\type)$. The regulator chooses the
cutoff and floor jointly by solving
\begin{equation}
\max_{\type_E,K}\quad\mathcal J(\type_E,K)
:=\int_{\underline\type}^{\type_E} \left[S(q^c(\capa_K(\type),\type))+\alpha v_K(\type)\right]\dint F(\type)+(1-\alpha)F(\type_E)v_K(\type_E),
\label{F}\tag{F}
\end{equation}
subject to $\underline\type\leq\type_E\leq\overline\type$,
$K\geq q^{LF}(\type_E)$, and
$v_K(\type_E)=\pi(K,\type_E)\geq0$.
The integral measures consumer surplus and weighted gross firm profit among
entrants. The final term accounts for collecting the common fee: each dollar
received by the regulator reduces the entrant's profit by one dollar, giving
a net welfare gain of $1-\alpha$.

Holding the cutoff fixed, the fee is pinned down by the marginal entrant's
profit. A tighter floor reduces that profit and hence the fee paid by every
entrant. Holding the floor fixed, admitting a higher-cost marginal type
similarly requires a lower common fee.  
The floor structure follows from the same fixed-terminal comparison used in
the proof of \autoref{thm_capa_floor}. 
For a given cutoff and marginal entrant's gross profit, fee revenue is fixed. 
A floor can then improve the allocation among entrants while preserving that profit. 
The regulator still must choose the cutoff and
floor that give the best overall combination. In particular, the baseline
condition for underutilization must now account for forgone common fees.

\begin{corollary}[Underutilized capacity floor with an entry fee]\label{cor_entry_fee_underuse}
Suppose \eqref{FL} holds. Let $(\type_E^*,K^*)$ solve \eqref{F}, and define
$q_E:=q^c(\type_E^*)$ and
$a_E:=\type_E^*-\capacost'(q_E)$. If $a_E\geq\underline\type$ and
\begin{equation}
-q_EP'(q_E)>
\max\left\{
\frac{\capacost(q_E)}{q_E},
\frac{F(\type_E^*)-\alpha F(a_E)}{f(a_E)}
\right\},
\label{U-E}\tag{U-E}
\end{equation}
then $K^*>q_E$, and a positive measure of entrants underutilizes the floor.
If, in addition,
\begin{equation}
-q_EP'(q_E)<\frac{\capacost(q^{LF}(\underline\type))}{q_E},
\label{LF-E}\tag{LF-E}
\end{equation}
there are cutoffs
$\underline\type<\type_{FL}<\type_U<\type_E^*$ separating three
positive-length entrant regions: laissez faire, full utilization of the floor,
and underutilization. When $\type_E^*<\overline\type$, a fourth region
consists of nonentrants, with zero capacity and output.
\end{corollary}

\autoref{cor_entry_fee_underuse} is analogous to  \autoref{prop_underuse}. The first term in \eqref{U-E} ensures viability as the
floor passes the marginal type's continuation output. The second incorporates
the revenue lost when the floor lowers the fee paid by all entrants.
Condition \eqref{LF-E} preserves a nonempty laissez-faire region. Crucially,
these conditions are evaluated at the \emph{endogenous optimal cutoff} from \autoref{prop_entry_fee}.

\medskip

To sum up, the additional instruments considered in this section relax different constraints. Subsidies encourage investment
but cannot collect rents through taxes; a common fee collects cutoff rent
without changing capacity comparisons within the menu. In both cases, capacity
expands output only up to the firm's continuation ceiling.

\section{Conclusion}
\label{sec_conclusion}

Capacity regulation can expand service even when the regulator cannot contract on the firm’s
output. Its effect comes from the difference between an investment decision that accounts for
capacity cost and a production decision made after that cost is sunk. We provide conditions
under which a single capacity floor is optimal and every type continues to operate. Under
further conditions, the same floor is fully used by some types and knowingly underutilized
by others.
This outcome brings together an output benefit and a resource cost. Raising the floor increases service when capacity constrains production; beyond a firm’s continuation output,
it adds only installation cost. Private information prevents the regulator from treating these
consequences as independent choices for each type. An exemption that saves a slack firm’s
capacity cost also changes the menu available to other types. Optimal capacity control
must account for those incentive effects across types, and can entail physical waste despite improving on laissez-faire.
The extensions show that access to money does not settle this tradeoff by itself. Nonnegative capacity subsidies can support greater investment by efficient firms while leaving a
zero-subsidy slack-capacity tail.
A common entry fee screens participation but does not alter capacity
comparisons within a given menu. The relevant question is
which payments the regulator can condition on which choices.
Capacity requirements should therefore be evaluated jointly with output, costs, and participation. Observed underutilization alone does not imply excessive regulation; it can be a necessary evil of optimal capacity control.

\bibliography{references}
\bibliographystyle{ecta}

\pagebreak

\appendix

\section{Reformulation of the Mechanism Design}
\label{apx}
\begin{lemma}\label{lem:opt_ctrl}
    The direct mechanism design problem \eqref{direct} is equivalent to the following optimal control problem with piecewise smooth
    state variables $(\Pi, \capa)$ and piecewise continuous control $z$ with finitely many jumps: \begin{equation}
        \max_{\Pi, \capa, z} \intfullsupport[V(q^c(\capa(\type), \type)) -q^c(\capa(\type), \type) P(q^c(\capa(\type), \type)) + \alpha \Pi(\type)] f(\type) \dint \type \label{controlProblem}
    \end{equation}
subject to, for almost all $\type \in \typespace$, \begin{align}
    &\Pi'(\type) = -q^c(\capa(\type), \type) \label{IC-Env} \tag{IC-Env}\\
    &\capa'(\type) = z(\type) \leq 0 \label{IC-Mon} \tag{IC-Mon}\\
    &\Pi(\type) = \pi(\capa(\type), \type) \label{N-OCT} \tag{N'}
\end{align}
with boundary conditions \[
\Pi(\underline{\type}) \leq \pi(q^{LF}(\underline{\type}),\underline{\type}),\; \Pi(\overline{\type}) \geq 0, \; \capa(\underline{\type}) \text{ free}, \; \capa(\overline{\type}) \geq 0.
\]
At every interior jump $\sigma \in (\underline{\type}, \overline{\type})$, the states satisfy\[
\capa(\sigma_-) \geq \capa(\sigma) \geq \capa(\sigma_+) \text{ and } \Pi(\sigma) = \pi(\capa(\sigma_-), \sigma) = \pi(\capa(\sigma), \sigma) = \pi(\capa(\sigma_+), \sigma).
\]
\end{lemma}

\begin{proof}[Proof of \autoref{lem:opt_ctrl}] For every fixed capacity $\capa$, uniqueness of continuation output and the envelope theorem give, for $\type_1 < \type_2$, \begin{equation}
    \pi(\capa, \type_2) - \pi(\capa, \type_1) = - \int_{\type_1}^{\type_2} q^c(\capa, t) \dint t. \label{env_lem_oc}
\end{equation}
Thus, $\frac{\partial \pi(\capa, \type)}{\partial \type} = - q^c(\capa, \type)$ throughout the interior of the type space.

Consider first a feasible direct mechanism. Under \eqref{N}, we show that \eqref{IC} is equivalent to the monotone condition \eqref{IC-Mon} and the envelope condition \eqref{IC-Env}. We first establish monotonicity. Suppose that $\type_1 < \type_2$ but $\capa_1:= \capa(\type_1) < \capa(\type_2) =: \capa_2$. The two incentive constraints involving $\type_1$ and $\type_2$ imply \[
\pi(\capa_2, \type_1) - \pi(\capa_1, \type_1) \leq 0 \leq \pi(\capa_2, \type_2) - \pi(\capa_1, \type_2).
\]
 The difference on the right minus the difference on the left is \[
- \int_{\type_1}^{\type_2} [q^c(\capa_2, t) - q^c(\capa_1, t)] \dint t \geq 0,
\] 
using \eqref{env_lem_oc}.
Recall that $q^c(\capa, t)  = \min \{\capa, q^c(t)\}$, which is nondecreasing in $\capa$. Since $\capa_2 > \capa_1$, $q^c(\capa_2,t) \geq q^c(\capa_1,t)$ for all $t \in [\type_1, \type_2]$. This implies the left-hand side of the above inequality should be nonpositive, i.e., \[
-\int_{\type_1}^{\type_2} [q^c(\capa_2, t) - q^c(\capa_1, t)] \dint t = 0.
\]
Again, since $q^c(\capa_2,t) \geq q^c(\capa_1,t)$ for all $t \in [\type_1, \type_2]$, it follows that \[
q^c(\capa_2, t) \equiv q^c(\capa_1, t), \forall t \in [\type_1, \type_2],
\]
which implies $q^c(t) \leq \capa_1$ throughout $[\type_1, \type_2]$. In particular, at $\type_2$, both capacities $\capa_1$ and $\capa_2$ induce the same output $q^c(\type_2) \leq \capa_1$. Therefore,\begin{align*}
    \pi(\capa_2, \type_2) - \pi(\capa_1, \type_2) &= (q^cP(q^c) - \type_2\cdot q^c - \capacost(\capa_2))- (q^cP(q^c) - \type_2\cdot q^c - \capacost(\capa_1))\\
    &= - [\capacost(\capa_2) - \capacost(\capa_1)] < 0, 
\end{align*}
contradicting \eqref{IC}. Thus $\capa$ is nonincreasing. 

Next, we turn to \eqref{IC-Env}. Combining \eqref{IC} and \eqref{N}, we have the value-function representation \[
\Pi(\type) = \max_{\hat{\type}} \pi(\capa(\hat{\type}), \type).
\]
Since $\pi(\capa, \type)$ is continuously differentiable in $\type$ for all $\capa \geq 0$, \cite{MilgromSegal02Env}'s envelope theorem implies \[
\Pi(\type) = \Pi(\overline{\type}) + \int_\type^{\overline{\type}} q^c(\capa(t),t)\dint t,
\]
and $\Pi'(\type) = -q^c(\capa(\type),\type)$ wherever $\Pi(\type)$ is differentiable.

Conversely, we show that \eqref{IC-Env} and \eqref{IC-Mon} combined are sufficient for \eqref{IC}. Suppose an admissible mechanism satisfies the control constraints. Fix $\type_1< \type_2$ and write $\capa_i = \capa(\type_i)$. Evaluate
\begin{align*}
    \pi(\capa_1, \type_1) - \pi(\capa_2, \type_1) &= \Pi(\type_1) - \pi(\capa_2, \type_1) \\
    &=\Pi(\type_1) - \Pi(\type_2) + \pi(\capa_2, \type_2) - \pi(\capa_2, \type_1) \\
    &= \int_{\type_1}^{\type_2}[q^c(\capa(t),t) - q^c(\capa_2,t)] \dint t \geq 0,
\end{align*}
where the last equality uses \eqref{IC-Env} and \eqref{env_lem_oc}, and the last inequality holds because $\capa(t) \geq \capa_2$ by \eqref{IC-Mon}. A similar argument shows \[
 \pi(\capa_2, \type_2) - \pi(\capa_1, \type_2) \geq 0.
\]
These are the two pairwise incentive constraints. Since the pair of types is arbitrary, incentive compatibility is global, i.e., \eqref{IC} holds.

Finally, \eqref{IC-Env} and $q^c \geq 0$ make $\Pi$ nonincreasing. Consequently, \eqref{IR} is equivalent to $\Pi(\overline{\type}) \geq 0$. Together with monotonicity, the terminal capacity restriction ensures $\capa(\type) \geq 0$ for every type. 
However, $\Pi(\underline{\type})$ is bounded above by type-$\underline{\type}$'s laissez-faire profit $\pi(q^{LF}(\underline{\type}),\underline{\type})>0$ because no transfer is allowed.
The no-transfer identity makes the objectives in \eqref{direct} and \eqref{controlProblem} identical. \eqref{N-OCT} restates the no-transfer constraint \eqref{N}. This proves the claimed equivalence.
\end{proof}

\section{Proof of \autoref{thm_capa_floor}}
Write the consumer surplus as $S(q) := V(q) - qP(q)$. We solve the mechanism design problem by first fixing the terminal rent $\pi(\capa, \overline{\type})$ and establishing the optimality of a capacity floor, and then optimizing over the set of feasible terminal rents.
We proceed in three steps.  \textbf{Step 1} defines the fixed-terminal-rent problem and derives conditions to verify the optimality of a candidate mechanism. \textbf{Step 2} shows that under condition \eqref{FL}, a capacity floor implements an optimal mechanism that solves the fixed-terminal-rent problem. \textbf{Step 3} optimizes over the terminal rents and establishes the existence of the optimal capacity floor. 

\medskip
\noindent{\bf Step 1: Fixed terminal optimum verification.} 
In this step, we identify a set of sufficient conditions (i.e., \eqref{c1},\eqref{c2},\eqref{c3}, and \eqref{IC-Mon}) for a candidate mechanism to be optimal under the mechanism design problem with a fixed terminal value. Specifically,  let the terminal profit $r$ be such that $0 \leq r\leq \pi(q^{LF}(\overline{\type}), \overline{\type})$. Consider the following fixed-terminal problem:
\begin{equation}
    \max_{k, \Pi} J(k,\Pi)
=\int_{\underline\type}^{\overline{\type}}
[S(q^c(\capa(\type),\type))+\alpha\Pi(\type)]f(\type) \dint \type, \label{T}\tag{T}
\end{equation}
with nonnegative, piecewise-continuous $\capa$ and an absolutely continuous profit $\Pi$ on $[\underline{\type}, \overline{\type}]$ satisfying 
\[
\Pi(\overline{\type}) = r, \; \Pi'(\type) = -q^c(\capa(\type), \type) \text{ a.e., }\; \Pi(\type) = \pi(\capa(\type),\type)
\]
and \eqref{IC-Mon}. Consider a relaxation of Problem \eqref{T} by dropping \eqref{IC-Mon}, and call this relaxed fixed-terminal problem (R).
Then, every admissible IC mechanism with terminal profit $r$ is feasible under the relaxed problem (R). 

Let the \textit{costate of} $\Pi$,  $\Lambda\colon [\underline{\type}, \overline{\type}] \to \R$, be an absolutely continuous function. Let the \textit{multiplier of \eqref{N-OCT}}, $\lambda\colon [\underline{\type}, \overline{\type}] \to \R$, be an integrable function. 
Pick $\Lambda, \lambda$ such that \begin{equation}
    \Lambda' = \lambda - \alpha f. \label{c1}\tag{c1}
\end{equation}
Define the Hamiltonian as
\[
H(\type, \capa, \Pi) = f(\type)[S(q^c(\capa, \type)) + \alpha \Pi] - \Lambda (\type) q^c(\capa, \type) + \lambda(\type)[\pi(\capa, \type) - \Pi].
\]
Suppose there is a feasible pair $(\capa^*, \Pi^*)$ such that 
\begin{equation}
    \capa^* \text{ globally maximizes } H(\type, \cdot , \Pi^*(\type)) \text{ almost everywhere}, \label{c2}\tag{c2}
\end{equation}
and 
\begin{equation}
    \Lambda(\underline{\type})\cdot [\Pi(\underline\type)-\Pi^*(\underline\type)]\leq 0 \text{ for every feasible candidate } \Pi. \label{c3}\tag{c3}
\end{equation}
Along every feasible path, $\pi(\capa, \type) = \Pi$ and $\Pi' = -q^c$ imply a reformulation of the objective: \[
J(\capa,\Pi)
=\int_{\underline\type}^{\overline{\type}}
[H(\type,\capa,\Pi)-\Lambda(\type)\Pi'(\type)]\dint \type.
\]
Moreover, $H$ is affine in $\Pi$ with $\frac{\partial H}{\partial\Pi}  = \alpha f - \lambda = -\Lambda'$. Using integration by parts and the common terminal value, we write the difference between the objective value under another feasible mechanism\footnote{
This consideration includes the mechanisms with shutdown, i.e., $\capa(\type) = \Pi(\type) = 0$ for some $\type$.
} $(\capa, \Pi)$ and the current candidate $(\capa^*,\Pi^*)$ as 
\[
J(\capa,\Pi)-J(\capa^*,\Pi^*)
=\int_{\underline\type}^{\overline{\type}}
[H(\type,\capa,\Pi^*)-H(\type,\capa^*,\Pi^*)]\dint\type+\Lambda(\underline\type)
[\Pi(\underline\type)-\Pi^*(\underline\type)].
\]
Both terms on the right-hand side are nonpositive, so the candidate solves Problem (R). If $\capa^*$ is nonincreasing, then \autoref{lem:opt_ctrl} supplies global IC, and thus it solves Problem \eqref{T}.

\medskip

\noindent{\bf Step 2: Under \eqref{FL}, a capacity floor solves every fixed-terminal problem.}
The type-$\overline{\type}$ firm's profit is continuous, has its unique maximizer at $q^{LF}(\overline{\type})$, is strictly decreasing above that maximizer, and tends to negative infinity. 
Hence, there is a unique upper root $K = K(\overline{\type}, r) \geq q^{LF}(\overline{\type})$ that solves \[
\pi(K, \overline{\type}) = r.
\]
According to Appendix \ref{apx:q-m-q-lf}, $q^{LF}(\overline{\type}) > 0$, so $K > 0$. Consider a candidate mechanism $(\capa^*, \Pi^*)$ given by \[
\capa^*(\type) = \max\{K, q^{LF}(\type)\},\; \Pi^*(\type) = \pi(\capa^*(\type), \type).
\]
The capacity is nonincreasing, satisfying \eqref{IC-Mon}. The profit is continuous and piecewise smooth with $(\Pi^*)' = -q^c$, and has terminal value $r$. It is globally IC by \autoref{lem:opt_ctrl}. Note that $\capa^*$ is outcome equivalent to a capacity floor $\{0\} \cup [K, \infty)$. Next, we show that under \eqref{FL}, $k^*$ pointwise maximizes the Hamiltonian $H(\type, \cdot, \Pi^*(\type))$ for each $\type$.

There are two cases. 
Suppose, first, that $K \geq q^c(\underline{\type})$. Then every type $\type \in \typespace$ weakly underutilizes capacity ($q \leq \capa$). Set \[
\lambda = 0, \; \Lambda(\type) = -\alpha [F(\type) - F(\underline{\type})].
\]
Then $\Lambda(\underline{\type}) = 0$, satisfying \eqref{c1} and \eqref{c3}. For each $\type$, substituting $\Lambda$ and $\lambda$ gives the Hamiltonian \[
H_{\type}(\capa) = f(\type) S(q^c(\capa, \type)) +\alpha[F(\type) - F(\underline{\type})]q^c(\capa, \type) + \alpha f(\type)\Pi^*(\type).
\]
For the tight branch ($0 \leq \capa \leq q^c(\type)$), capacity binds, so $q^c(\capa, \type) = \capa$. Hence, for $\capa_2 > \capa_1$ on the tight branch, \[
H_{\type}(\capa_2)  - H_{\type}(\capa_1)  = f(\type)[S(\capa_2) - S(\capa_1)] + \alpha [F(\type) - F(\underline{\type})](\capa_2 - \capa_1) \geq 0,
\]
because both terms are nonnegative. In particular, $S'(k) = -kP'(k) > 0$ for $k > 0$. Therefore, the Hamiltonian increases in $\capa$ as $\capa$ approaches $q^c(\type)$. For the slack branch ($\capa \geq q^c(\type)$), $q^c(\capa, \type) = q^c(\type)$. Increasing capacity no longer changes output. Moreover, the only other place where $\capa$ enters the Hamiltonian is through $\lambda(\type)\pi(\capa, \type)$, which vanishes because $\lambda = 0$. Therefore, the Hamiltonian is constant in $\capa$ throughout the slack branch. Consequently, every $\capa \geq q^c(\type)$ is a global pointwise maximizer. Finally, $q^c(\type)$ decreases with $\type$, so $K \geq q^c(\underline{\type}) \geq q^c(\type)$ for every type. Therefore, the proposed floor $K$ pointwise maximizes each type's Hamiltonian, \textit{given that} the profit $\Pi^*(\type)$ satisfies terminal value $\Pi^*(\overline{\type}) = r$. Condition \eqref{c2} is satisfied.

Now suppose that $0 < K < q^c(\underline{\type})$, i.e., at least some types will fully utilize the capacity ($q = \capa$). Define cutoffs \[
a := (KP(K))' - \capacost'(K), \; b := (KP(K))', \; d := \max \{a, \underline{\type}\}.
\]
Note that $a \leq \overline{\type}$, $a < b$, and $b > \underline{\type}$. 
Call the interval $[\underline{\type}, d]$ the \emph{laissez-faire region}, the interval $(d,b]$ the \emph{fully-used-floor region}, and the interval $(b,\overline{\type}]$ the \emph{underutilization region}. 
Note that the laissez-faire region and the underutilization region may be empty. Nevertheless, they are covered by the following analysis. 

Define output induced by a candidate floor $K$ as \[
Q_K(\type):= \min \{\max\{K, q^{LF}(\type)\}, q^c(\type)\},
\]
and for the same $K$,\[
\Phi_K(\type):= f(\type)S'(Q_K(\type)) + \alpha F(\type).
\]
We claim that $\Phi_K$ is nondecreasing. To see this, note that $Q_K(\type) \geq q^{LF}(\overline{\type}) =: \hat{q}$. Since $S''(q) \leq 0$, we have \[
S'(Q_K(\type)) \leq S'(\hat{q}).
\]
Moreover, $Q_K$ is nonincreasing. We have\[
\Phi_K' = f'S'(Q_K) +  \underbrace{ fS''(Q_K) Q_K'}_{\geq 0} + \alpha f, a.e.
\]
If $f' \geq 0$, then $f'S'(Q_K) + \alpha f \geq 0$. If $f'< 0$, then $S'(Q_K(\type)) \leq S'(\hat{q})$ implies \[
f'S'(Q_K) + \alpha f \geq f'S'(\hat{q}) + \alpha f \geq 0
\]
by \eqref{FL}. Hence $\Phi_K'(\type) \geq 0$.

We propose the following costate $\Lambda$ and multiplier $\lambda$ consistent with \eqref{c1}: if $d = a$, \[
\Lambda(\type) = \begin{cases}
    f(\type) S'(q^{LF}(\type)), & \type \in [\underline{\type}, d]\\
    \frac{\int_d^\type f(t)[S'(K)-\alpha(t-a)]dt}{\type-a}, & \type \in (d, b]\\
    \Lambda (b) - \alpha[F(\type) - F(b)], &\type \in (b, \overline{\type}]
\end{cases}, \quad \lambda(\type) = \begin{cases}
    \Lambda'(\type) + \alpha f(\type), & \type \in [\underline{\type}, d]\\
    \frac{f(\type)S'(K)-\Lambda(\type)}{\type-a}, & \type \in (d, b]\\
    0,  &\type \in (b, \overline{\type}]
\end{cases}
\]
with $\Lambda(a) = f(a)S'(K)$. If $d > a$, then $d= \underline{\type}$ and the laissez-faire ``region'' is a singleton. In this case, set $\Lambda(\underline{\type}) = 0$, and use the same formula as in the case $d = a$ for other regions. 

Note that the underutilization region, if it exists, shares \(\lambda=0\) and \(\Lambda'=-\alpha f\) with the previous case $K \geq q^c(\underline\type)$.

We first verify that $\lambda(\type) \geq 0$, for almost every $\type \in [\underline\type, \overline{\type}]$. Note that $S' > 0$ and $(q^{LF})' < 0$. By \eqref{FL}, we have $S'' \leq 0$.
On the laissez-faire region $[\underline{\type}, d]$, $Q_K = q^{LF}$. Therefore, 
 \[
\lambda(\type) = \frac{d}{d\type} [f(\type)S'(q^{LF}(\type)) + \alpha F(\type)] = \Phi_K'(\type) \geq 0.
\]

On the fully-used-floor region,  $Q_K(\type) \equiv K$. Define $h_K(\type):= S'(K)f(\type)  + \alpha[F(\type) - F(d)]$. Then, \[
h_K(\type) = \Phi_K(\type) - \alpha F(d).
\]
Therefore, on the floor region, 
\[
\lambda(\type)
= \frac{h_K(\type)}{\type - a} - \frac{1}{(\type - a)^2} \int_d^\type h_K(t) \dint t.
\]
When \(d=a\) and the floor region is nonempty, fix $0 < \epsilon< \min\{b, \overline{\type}\} - a$. We show that $\lambda(\type)$ is integrable near $\type = a$.
Absolute continuity of $h_K$ and integration by parts give, for $x > a$, \[
\lambda(x) = \frac{1}{(x-a)^2}\int_a^x (t-a) \dint h_K(t)  \geq 0.
\]
Since $h_K'\geq 0$ almost everywhere, Tonelli's theorem\footnote{
See \citet[Theorem~9.2]{Iyer2022MeasureTheory}.
} implies \[
\int_a^{a+\epsilon} \lambda(x)\dint x = \int_a^{a+\epsilon} \left(1 - \frac{t - a}{\epsilon}\right)\dint h_K(t) 
\leq h_K(a+\epsilon) - h_K(a) < \infty.
\]
Thus $\lambda(\type)$ is locally integrable in the right neighborhood of $\type = a$.  The absolute continuity of $\Lambda$ is thus also established.
Because $h_K(\type)$ is nondecreasing,
\[
\int_d^\type h_K(t) \dint t \leq (\type - d) h_K(\type).
\]
Therefore, \[
\lambda(\type) \geq \frac{d-a}{(\type - a)^2} h_K(\type) \geq 0.
\]
The last inequality is implied by $h_K(\type) \geq 0$. To see this, note that \[
h_K(d) = S'(K)f(d) + \alpha [F(d) - F(d)] = S'(K)f(d) > 0,
\]
and that $h_K$ is nondecreasing.

Finally, on the underutilization region, $\lambda(\type) = 0$.

We are now ready to show that the capacity schedule $\capa^*$ pointwise maximizes the Hamiltonian. On a tight branch ($0 \leq \capa \leq q^c(\type)$), \[
\frac{\partial H}{\partial \capa} = fS'(\capa) - \Lambda + \lambda[(\capa P(\capa))' - \type - \capacost'(\capa)],
\]
and 
\[
\frac{\partial^2 H}{\partial \capa^2} = fS''(\capa) +\lambda [(\capa P(\capa))'' - \capacost''(\capa)] \leq 0.
\]
Therefore, the Hamiltonian is concave in $\capa$. 
The proposed costate and multiplier make $\frac{\partial H}{\partial \capa} = 0$ at $\capa = q^{LF}(\type)$ in the laissez-faire region and at $\capa = K$ in the fully-used-floor region. To see this, note that in the laissez-faire region, $q^{LF}$ solves $(kP(k))' - \type - \capacost'(\capa) = 0$. In the fully-used-floor region, we use the definition $a = (KP(K))' - \capacost'(K)$. Then, it is straightforward to verify that the choice of $\Lambda$ and $\lambda$ deliberately makes the first-order condition ($\frac{\partial H}{\partial \capa} = 0$) hold. Thus, $\capa^*$ pointwise maximizes $H(\type, \cdot, \Pi^*(\type))$ on the tight branch.

On a slack branch ($\capa \geq q^c(\type)$), we have \[
\frac{\partial H}{\partial \capa} = -\lambda \capacost'(\capa) \leq 0.
\]
Among slack capacities, \(k=q^c(\theta)\) maximizes the Hamiltonian.
For a type whose candidate allocation is tight, that allocation already maximizes the Hamiltonian over the tight branch. Its value is therefore at least the value at \(q^c(\type)\), and no slack capacity can improve on it.

If $b \geq \overline{\type}$, there are no strictly underutilizing types and the argument is complete. Suppose that $b < \overline{\type}$, i.e., there is a positive measure of underutilizing types.
At $\type=b$, we have $q^c(b)=K$ and $b-a=\capacost'(K)$.
The tight-side derivative satisfies
\[
H_{\capa}(b,K^-)
=f(b)S'(K)-\Lambda(b)-\lambda(b)\capacost'(K)=0,
\]
because the multiplier construction gives
\[
f(b)S'(K)-\Lambda(b)
=(b-a)\lambda(b)
=\capacost'(K)\lambda(b).
\]
On the slack side,
\[
H_{\capa}(b,K^+)
=-\lambda(b)\capacost'(K)\leq0.
\]
Thus $K=q^c(b)$ is a kink maximizer. Moreover,
$\Lambda(b)\leq f(b)S'(K)$, an inequality used below for
the strictly underutilized types.
For $\type > b$, the firm is strictly underutilizing the capacity, and we have $\lambda(\type) \equiv 0$. Then, \[
\Lambda(\type) = \Lambda(b) - \alpha[F(\type) - F(b)],
\]
which implies \[
\Lambda(\type) + \alpha F(\type) = \Lambda(b) + \alpha F(b).
\]
Meanwhile, monotonicity of $\Phi_K$ gives \begin{align*}
    f(\type) S'(q^c(\type)) + \alpha F(\type) &= \Phi_K(\type) \\
    &\geq \Phi_K(b) \\
    &=f(b)S'(K) + \alpha F(b)\\
    &\geq  \Lambda(b) + \alpha F(b) = \Lambda(\type) + \alpha F(\type).
\end{align*}
That is, $f(\type) S'(q^c(\type)) \geq \Lambda(\type)$.
 Hence the Hamiltonian increases in $\capa$ until $\capa = q^c(\type)$, and remains constant afterward. Therefore, $K$, which keeps the types choosing $q(\type)\leq K$, remains a global maximizer. 
Thus the candidate floor $K$ pointwise maximizes the Hamiltonian. To sum up, we have established \eqref{c2}, i.e., that $\capa^*$ pointwise maximizes $H(\type, \cdot, \Pi^*)$ for each $\type \in [\underline{\type}, \overline{\type}]$, with given $\Lambda$ and $\lambda$.

Finally, we verify that the transversality condition \eqref{c3} holds. If $a < \underline{\type}$, then $\Lambda (\underline{\type}) = 0$ by construction. If $a \geq \underline{\type}$, the lowest type receives its profit-maximizing capacity, so any other candidate mechanism has $\Pi(\underline{\type}) \leq \Pi^*(\underline{\type})$, while $\Lambda(\underline{\type}) \geq 0$. Condition \eqref{c3} thus follows. 

In conclusion, we have established the optimality of $(\capa^*, \Pi^*)$ for Problem \eqref{T} with terminal condition $\Pi^*(\overline{\type}) = r$.

\medskip

\noindent{\bf Step 3: Optimization of terminal profit.}
Now we consider the original mechanism design problem \eqref{controlProblem}. 
For any feasible mechanism, let $r = \Pi(\overline{\type})$. The no-transfer constraint, IR, and pointwise optimality of $q^{LF}(\overline{\type})$ imply $0 \leq r \leq \pi(q^{LF}(\overline{\type}), \overline{\type})=: \overline{r}$.
\textbf{Step 2} replaces the mechanism by a weakly better floor with the same terminal profit.
 Define $g(K):= \pi(K, \overline{\type})$. The function $g$ is continuous and strictly decreasing on the upper segment $K \geq q^{LF}(\overline{\type})$. 
Moreover, $g(q^{LF}(\overline{\type})) = \pi(q^{LF}(\overline{\type}), \overline{\type}) > 0$, while the unbounded capacity cost implies $g(K) \rightarrow -\infty$ as $K \rightarrow \infty$. Hence there is a unique upper zero profit root $\overline{K}$ satisfying \[
\overline{K} > q^{LF}(\overline{\type}), \; g(\overline{K}) = 0.
\]
Thus $g$ maps $[q^{LF}(\overline{\type}), \overline{K}]$ continuously and strictly decreasingly onto $[0, \overline{r}]$. Its inverse $K(r)$ is therefore continuous on the entire closed interval $[0, \overline{r}]$. For each $r$, define the induced floor allocation by \[
\capa_r(\type) = \max\{K(r), q^{LF}(\type)\}, \; \Pi_r(\type) = \pi(\capa_r(\type), \type).
\]
Its welfare is \[
J(r) = \intfullsupport[S(q^c(\capa_r(\type),\type)) + \alpha \Pi_r(\type)] \dint F(\type).
\]
For each $\type$, the maps $r \mapsto \capa_r(\type),q^c(\capa_r(\type),\type)$, and $\Pi_r(\type)$ are continuous because $K(r)$ is continuous. They are also uniformly bounded. Dominated convergence implies that $J(r)$ is continuous. Since $[0, \overline{r}]$ is compact, the Weierstrass extreme value theorem gives a maximizer $\hat{r}$. Let $\widehat{K}:= K(\hat{r})$. This establishes the existence of a globally optimal floor. 

Finally, we discuss the bounds of $\widehat{K}$. Suppose that $\widehat{K} > q^c(\underline\type)$. Then \[
\widehat{K} > q^c(\underline\type) \geq q^c(\type) > q^{LF}(\type)
\]
for every $\type$. Therefore, every type produces their continuation monopoly output $q^c(\type)$.
Consider lowering the floor from $\widehat{K}$ to $K_0 := q^c(\underline\type)$, which changes no output or consumer surplus. However, it reduces capacity cost by $\capacost(\widehat{K}) - \capacost(K_0) > 0$ pointwise. Hence an optimal capacity floor satisfies $K^* \leq q^c(\underline{\type})$. On the other hand, the lower bound $q^{LF}(\overline{\type})$ is implied by \autoref{lem:full-delegation}. This finishes the proof.
\hfill\qedsymbol

\section{Proof of \autoref{prop_underuse}}\label{apx:proof_underuse}
Write the consumer surplus as $S(q):= V(q) - qP(q)$, so $S'(q) = -qP'(q)$. The viability part of \eqref{U} implies\[
\pi(\overline{q}, \overline{\type}) = \overline{q}P(\overline{q}) - \overline{\type}\cdot \overline{q} - \capacost(\overline{q}) = -\overline{q}^2P'(\overline{q}) - \capacost(\overline{q}) > 0,
\]
where the second equality is obtained by substituting in the monopoly first-order condition for continuation output. Thus the top type $\overline{\type}$ participates strictly when the capacity floor equals $\overline{q}$. By continuity, type $\overline{\type}$ also participates when the capacity floor is in a small right neighborhood of $\overline{q}$.

Let $J(K)$ be welfare under a full-participation floor. For $K \in [q^{LF}(\overline{\type}), \overline{q}]$, all types have $q^c(\type) \geq \overline{q} \geq K$. Hence every type affected by the floor (i.e., their output differs from the laissez-faire one) fully uses it. For such a type, \[
\frac{\partial W(K, K; \type)}{\partial K} = S'(K) + \alpha[(KP(K))' - \type - \capacost'(K)].
 \]
Since the revenue $qP(q)$ is strictly concave, we have \[
(KP(K))' \geq (\overline{q}P(\overline{q}))' = \overline{\type} \geq \type.
\] 
Condition \eqref{FL} makes $S'$ nonincreasing in capacity, while convexity makes $\capacost'$ nondecreasing. Therefore, \[
\frac{\partial W(K, K; \type)}{\partial K} \geq S'(K) - \alpha \capacost'(K) \geq S'(\overline{q}) - \alpha \capacost'(\overline{q}) > 0.
\]
In particular, the welfare is continuous when the floor first reaches a type's laissez-faire capacity. Thus $J(K)$ is strictly increasing from $q^{LF}(\overline{\type})$ through $\overline{q}$.
Therefore, the full-participation capacity floor policies with $K< \overline{q}$ are welfare dominated.
Lower floors than $q^{LF}(\overline{\type})$ implement full delegation and are dominated as is shown in \autoref{lem:full-delegation}.

It remains to exclude $K = \overline{q}$. Using notation similar to that in the proof of \autoref{thm_capa_floor}, 
define the type cutoffs across laissez-faire region, fully-used-floor region, and underutilization region as functions of $K$: \[
a(K) := (KP(K))' - \capacost'(K),\; b(K) := (KP(K))',\; d(K) = \max\{\underline{\type}, a(K)\}.
\]
Then, $[\underline{\type}, d(K)]$ is the laissez-faire region, $(d(K), b(K)]$ is the fully-used-floor region, and $(b(K), \overline{\type}]$ is the underutilization  region.

Consider a floor $K$ in a sufficiently small right neighborhood of $\overline{q}$. The floor welfare is \[
J(K) = \int_{\underline{\type}}^{d(K)}W_{LF}(\type)\dint F(\type) + \int_{d(K)}^{b(K)}W_{F}(K,\type)\dint F(\type) + \int_{b(K)}^{\overline{\type}}W_{O}(K,\type)\dint F(\type),
\]
where we denote $W_{LF}(\type) := W(q^{LF},q^{LF};\type)$ if $\type$ is in the laissez-faire region; $W_{F}(K,\type) := W(K,K; \type)$ if $\type$ is in the fully-used-floor region; and $W_O(K,\type) = W(q^c, K; \type)$ if $\type$ is in the underutilization region. 
Note that the ex post welfare is continuous.
 $W_{LF}$ does not depend on $K$. The direct derivatives of the other integrands are, respectively, \[
\frac{\partial W_{F}}{\partial K} = S'(K) + \alpha [(KP(K))' - \type - \capacost'(K)] = S'(K) + \alpha [a(K) - \type], 
\]
and \[
\frac{\partial W_{O}}{\partial K} = -\alpha \capacost'(K),
\]
because the output $q^c(\type)$ is unaffected by $K$ in the underutilization region. Since $W(q(\type),\capa(\type), \type)$ is continuous across $\type$, the right-derivative of $J(K)$ at $\overline{q}$ is then \begin{align*}
    J'(\overline{q}_+) &= \lim_{K \downarrow \overline{q}}\left\{\int_{d(K)}^{b(K)} [S'(K) + \alpha(a(K) - \type)] \dint F(\type) - \alpha \capacost'(K)[1-F(b(K))]\right\} \\
    &=  \int_{d_m}^{\overline{\type}} [S'(\overline{q}) + \alpha (a_m - \type)]f(\type)\dint \type,
\end{align*}
where $a_m := a(\overline{q}) = \overline{\type} - \capacost'(\overline{q})$ and $d_m = \max \{a_m, \underline{\type}\}$. Note that $b(\overline{q}) = \overline{\type}$, and thus $\lim_{K \downarrow \overline{q}} \{1 - F(b(K))\} = 0$.
The interval $[d_m, \overline{\type}]$ has positive length because $a_m < \overline{\type}$. For every type in it, \[
\type - a_m \leq \overline{\type} - a_m = \capacost'(\overline{q}),
\]
so its integrand is bounded below by \[
S'(\overline{q}) - \alpha \capacost'(\overline{q}) > 0.
\]
We conclude that $J'(\overline{q}_+) >0$. Strict top-type participation makes the rightward perturbation feasible. Therefore, no welfare-maximizing full-participation floor can satisfy $K \leq \overline{q}$. 

By \autoref{thm_capa_floor}, there is a globally optimal full-participation floor. Select one with $K^* > \overline{q}$. Strict concavity of revenue gives \[
[(K^*)P(K^*)]' < [\overline{q}P(\overline{q})]' = \overline{\type}.
\]
Hence every type $\type$ in\[
(\max\{\underline{\type}, [(K^*)P(K^*)]'\}, \overline{\type}]
\]
has $q^c(\type) < K^*$, installs the smallest approved positive capacity $K^*$, and produces $q^c(\type)$. Because $f > 0$, the interval has positive probability. This proves the emergence of underutilization in the optimal mechanism. 

Suppose now that \eqref{LF} holds. Let $\overline{K}$ be type-$\overline{\type}$'s upper zero-profit floor. Then \[
\capacost(\overline{K}) = -\overline{q}^{2} P'(\overline{q}).
\]
Combining \eqref{LF} and \eqref{U}, we have \[
\overline{q} < \overline{K} < q^{LF}(\underline{\type}),
\]
where the second inequality is implied by $\capacost(\overline{K}) < \capacost(q^{LF}(\underline{\type}))$, a result of \eqref{LF}. The proof of \autoref{thm_capa_floor} has established an optimal floor $K^*$ satisfying $\overline{q} < K^* \leq \overline{K}$. Therefore, \[
K^* < q^{LF}(\underline{\type}),
\]
i.e., the lowest type $\underline{\type}$ operates strictly above the capacity floor, in a laissez-faire region. Therefore, we can define two cutoffs as \[
\type_{FL} := a(K^*) =  [K^*P(K^*)]' - \capacost'(K^*),\; \type_{U} := b(K^*) = [K^*P(K^*)]'.
\]
Since $a(\cdot)$ is strictly decreasing, we guarantee that $\type_{FL} > \underline{\type}$. Meanwhile, 
\[
\type_U - \type_{FL} = \capacost'(K^*) > 0 \text{ and } \overline{\type} - \type_U = [\overline{q}P(\overline{q})]' - [K^*P(K^*)]' >0.
\]
The second inequality is a result of $K^* > \overline{q}$. Hence we have three positive-length regions: \[
q(\type) = \begin{cases}
    q^{LF}(\type) & \text{on } [\underline{\type}, \type_{FL}]\\
    K^* & \text{on } (\type_{FL}, \type_U]\\
     q^c(\type) & \text{on } (\type_U, \overline{\type}]
\end{cases}, \quad 
\capa(\type) = \begin{cases}
    q^{LF}(\type) & \text{on } [\underline{\type}, \type_{FL}]\\
    K^* & \text{on } (\type_{FL}, \overline{\type}]
\end{cases}
\]
The profit $\Pi$ satisfies the envelope condition \[
\Pi(\type) = \Pi(\overline{\type}) +  \int_\type^{\overline{\type}}q(t)\dint t,
\]
where $\Pi(\overline{\type}) = \pi(K^*, \overline{\type})\geq 0$.
This finishes the proof of \autoref{prop_underuse}. 
\hfill\qedsymbol

\section{Proof of \autoref{prop_subsidy}}
\label{apx:proof_subsidy}
We show that the ex post welfare can be written in terms of $\capa$ and firm profit $\Pi(\capa, \type) := \pi(\capa, \type) + \subsidy(\capa)$. Substituting \[
\subsidy(\capa) = \Pi - [q^c(\capa, \type)P(q^c(\capa, \type)) - \type q^c(\capa, \type) - \capacost(\capa)]
\]
into $w^S(k, \type)$ yields \[
V(q^c(\capa, \type)) - \type q^c(\capa, \type) - \capacost(\capa) - (1-\alpha)\Pi.
\]

Using the revelation principle, we can restrict attention to direct mechanisms $(\capa(\type), \Pi(\type))$, where $\capa(\type)$ is the recommended capacity and $\Pi(\type)$ is the truthful profit for type $\type$. 
We first reformulate problem \eqref{subsidyDesign} as the following direct mechanism design problem:\begin{equation}
    \max_{\capa(\cdot), \Pi(\cdot)} J^S(\capa, \Pi) := \intfullsupport [V(q^c(\capa(\type), \type)) - \type q^c(\capa(\type), \type) - \capacost(\capa(\type)) - (1-\alpha)\Pi(\type)] \dint F(\type)\label{subsidyDirectProblem}
\end{equation}
subject to the following constraints for all $\type \in \typespace$: \begin{align}
    &\Pi(\type)\geq
 \Pi(\hat\type)
 +\pi(\capa(\hat\type),\type)
 -\pi(\capa(\hat\type),\hat\type), \forall \hat{\type} \in \typespace \tag{IC-S}\\
    &\Pi(\type) \geq 0 \tag{IR-S}\\
    &\Pi(\type) \geq \pi(\capa(\type), \type), \tag{NT}
\end{align}
where (NT) is the no-tax constraint.

We proceed in three steps. In \textbf{Step 1}, we verify that the proposed mechanism in \autoref{prop_subsidy} is feasible and implements the claimed allocation. \textbf{Step 2} finds an upper bound of the objective value of problem \eqref{subsidyDirectProblem} by replacing the actual capacity cost with its lower bound and by ignoring the pairwise IC constraints except for the local ones.  
\textbf{Step 3} shows that the proposed mechanism in \autoref{prop_subsidy} attains the upper bound under the sufficient conditions \eqref{FL}, \eqref{R3IR}, and \eqref{R3L}, and thus is optimal. 

\medskip

\noindent{\bf Step 1: Implementation of the proposed allocation.}
Consider the capacity floor $\{0\}\cup[K_\alpha, \infty)$ with subsidy scheme\[
\subsidy(\capa) = \begin{cases}
    \capacost(\capa)-\capacost(K_\alpha), & \text{if } K_\alpha \leq \capa \leq q^c(\underline{\type})\\
    0 &\text{ otherwise}
\end{cases}.
\]
Because $\capacost$ is increasing, every subsidy offered by this menu is
nonnegative. If type $\type$ selects a reimbursed capacity $\capa \in [K_\alpha, q^c(\underline{\type})]$, its payoff
is
\begin{align*}
& q^c(\capa, \type)P(q^c(\capa, \type))
-\type q^c(\capa, \type)
-\capacost(\capa)
+[\capacost(\capa)-\capacost(K_\alpha)]\\
&\qquad
=
q^c(\capa, \type)P(q^c(\capa, \type))
-\type q^c(\capa, \type)
-\capacost(K_\alpha).
\end{align*}
Thus the subsidy holds the firm's net capacity expenditure fixed at
$\capacost(K_\alpha)$. The remaining expression is maximized by attaining unconstrained continuation 
output $q^c(\type)$. Any capacity above $q^c(\underline\type)$ receives no subsidy and
induces the same output as $q^c(\underline\type)$ at a strictly higher net
installation cost. Such capacities are strictly dominated. Among operating
best responses, we select the smallest capacity, i.e., 
\[
\capa^*(\type)=\max\{q^c(\type),K_\alpha\}.
\]
It follows that
\[
q^*(\type)=q^c(\type),
\]
and the monopoly profit is $\Pi^*(\type) = q^c(\type)P(q^c(\type)) - \type q^c(\type) - \capacost(K_\alpha)$.

The monopoly-profit envelope and 
condition \eqref{R3IR} imply
\[
\Pi^*(\type)
\geq
\overline{q}P(\overline{q}) - \overline{\type} \overline{q} -\capacost(K_\alpha)
\geq 0
\]
for every type. Hence participation is optimal. 
Finally,
\[
\subsidy^*(\type)
=
\capacost(\capa^*(\type))-\capacost(K_\alpha)
=
\left[
\capacost(q^c(\type))-\capacost(K_\alpha)
\right]_+.
\]
This proves feasibility and  global incentive compatibility of the proposed mechanism.

\medskip

\noindent{\bf Step 2: An upper bound for every feasible mechanism.}
Consider any globally incentive-compatible mechanism
$(\capa(\cdot),\Pi(\cdot))$, and write
\[
x(\type):=q^c(\capa(\type),\type).
\]
The fixed-capacity envelope theorem and global incentive compatibility imply
\[
\Pi'(\type)=-x(\type)
\qquad\text{almost everywhere}.
\]
Moreover,
\[
0\leq x(\type)\leq q^c(\type),\; 
\Pi(\type)\geq\pi(\capa(\type),\type),
\]
where the second inequality is the no-tax constraint. We use these only as
necessary conditions and hence as a relaxation of the original mechanism
problem.

Split welfare at $\type_\alpha$:
\[
J^S(\capa,\Pi)=J_L+J_H.
\]
On the lower interval $[\underline{\type}, \type_\alpha]$, $\capa(\type)\geq x(\type)$ and $\capacost$ is
increasing, so
\[
\capacost(x(\type)) \leq \capacost(\capa(\type)).
\]
Meanwhile, integration by parts using the envelope condition gives
\begin{align*}
\int_{\underline{\type}}^{\type_\alpha}
\Pi(\type)f(\type)\dint\type
&=
F(\type_\alpha)\Pi(\type_\alpha)
+\int_{\underline{\type}}^{\type_\alpha}
F(\type)x(\type)\dint\type \notag\\
&=
\alpha\Pi(\type_\alpha)
+\int_{\underline{\type}}^{\type_\alpha}
F(\type)x(\type)\dint\type.
\end{align*}
It follows that
\begin{align}
J_L
\leq{}
\int_{\underline{\type}}^{\type_\alpha}
f(\type)
\left[
V(x(\type))
-\type x(\type)
-\capacost(x(\type))
-(1-\alpha)\frac{F(\type)}{f(\type)}x(\type)
\right]\dint\type
-\alpha(1-\alpha)\Pi(\type_\alpha).
\label{eq:r3-lower-bound}
\end{align}
On the upper interval $[\type_\alpha, \overline{\type}]$, 
we use the following equivalent expression for ex post welfare: \[
S(x(\type)) + \pi(\capa(\type),\type)-(1-\alpha)\Pi(\type),
\]
because the total gains from trade are $S(x(\type)) + \pi(\capa(\type),\type) = V(x(\type)) - \type x(\type) - \capacost(\capa(\type))$.
The no-tax constraint (NT) implies an upper bound of the terms
\[
\pi(\capa(\type),\type)-(1-\alpha)\Pi(\type)
=
\alpha\Pi(\type)
-\left[
\Pi(\type)-\pi(\capa(\type),\type)
\right]
\leq\alpha\Pi(\type).
\]
Also, integration by parts gives
\begin{align*}
\int_{\type_\alpha}^{\overline{\type}}
\Pi(\type)f(\type)\dint\type
&=
[1-F(\type_\alpha)]\Pi(\type_\alpha)
-\int_{\type_\alpha}^{\overline{\type}}
[1-F(\type)]x(\type)\dint\type
\notag\\
&=
(1-\alpha)\Pi(\type_\alpha)
-\int_{\type_\alpha}^{\overline{\type}}
[1-F(\type)]x(\type)\dint\type.
\end{align*}
Therefore,
\begin{align}
J_H
\leq
\int_{\type_\alpha}^{\overline{\type}}
f(\type)
\left[
S(x(\type))
-\alpha\frac{1-F(\type)}{f(\type)}x(\type)
\right]\dint\type
+\alpha(1-\alpha)\Pi(\type_\alpha).
\label{eq:r3-upper-bound}
\end{align}

Adding \eqref{eq:r3-lower-bound} and
\eqref{eq:r3-upper-bound} cancels the two terms involving
$\Pi(\type_\alpha)$. This is exactly where
$F(\type_\alpha)=\alpha$ is used. Thus every feasible no-tax mechanism
satisfies
\begin{align}
J^S(\capa,\Pi)
\leq{}&
\int_{\underline{\type}}^{\type_\alpha}
f(\type)\Phi_L(\type,x(\type))\dint\type
+
\int_{\type_\alpha}^{\overline{\type}}
f(\type)\Phi_H(\type,x(\type))\dint\type,
\label{eq:r3-global-bound}
\end{align}
where
\begin{align*}
\Phi_L(\type,q)
&:=
V(q)-\type q-\capacost(q)
-(1-\alpha)\frac{F(\type)}{f(\type)}q,\\
\Phi_H(\type,q)
&:=
S(q)-\alpha\frac{1-F(\type)}{f(\type)}q.
\end{align*}

\medskip

\noindent{\bf Step 3: Pointwise maximization of the upper bound.}
For every type, feasible induced output satisfies
\[
0\leq q\leq q^c(\type).
\]

For any fixed $\type\leq\type_\alpha$, the function
$\Phi_L(\type,\cdot)$ is strictly concave because
\[
\frac{\partial^2\Phi_L(\type,q)}{\partial q^2}
=
P'(q)-\capacost''(q)<0.
\]
Using the monopoly first-order condition
$\frac{d[qP(q)]}{dq}|_{q = q^c(\type)}=\type$ and the identity
\[
P(q^c(\type))-\type
=
-q^c(\type)P'(q^c(\type))
=
S'(q^c(\type)),
\]
its derivative at the upper endpoint $q = q^c(\type)$ is
\[
\frac{\partial\Phi_L}{\partial q}
\bigl(\type,q^c(\type)\bigr)
=
S'(q^c(\type))
-\capacost'(q^c(\type))
-(1-\alpha)\frac{F(\type)}{f(\type)}
\geq0
\]
by \eqref{R3L}. Hence $q^c(\type)$ maximizes
$\Phi_L(\type,q)$ over $[0,q^c(\type)]$.

For any fixed $\type\geq\type_\alpha$,
\[
\frac{\partial^2\Phi_H(\type,q)}{\partial q^2}
=
S''(q)\leq0.
\]
Moreover, we claim that
\[
\frac{\partial\Phi_H}{\partial q}
\bigl(\type,q^c(\type)\bigr)
=
S'(q^c(\type))
-\alpha\frac{1-F(\type)}{f(\type)}
\geq0, 
\]
 and thus $q^c(\type)$ also maximizes
$\Phi_H(\type,q)$ over $[0,q^c(\type)]$. To see this, define \(G(\type)
:=
f(\type)S'(q^c(\type))
-\alpha[1-F(\type)].\) We only need to show that $G(\type) \geq 0$. Differentiating, \[
G'(\type)
=
f'(\type)S'(q^c(\type))
+
\underbrace{f(\type)S''(q^c(\type))(q^c)'(\type)}_{\geq 0}
+
\alpha f(\type).
\]
We have \(q^c(\type)\geq q^c(\overline{\type})
\geq q^{LF}(\overline{\type})=\hat q.\) By $S'' \leq 0$, we also have $S'(q^c(\type)) \leq S'(\hat{q})$.
Now consider the sign of $f'$. If $f' \geq 0$, then $f'S'(q^c) \geq 0$, implying $G' \geq 0$. On the other hand, if $f' < 0$, $f'S'(q^c) \leq 0$, but \(f'S'(q^c)\geq f'S'(\hat q)\) and \eqref{FL} imply \[
f'S'(q^c)+\alpha f\geq0.
\]
Hence, $G' \geq 0$, which implies \[
G(\type) \geq G(\type_\alpha) > 0, \forall \type \geq \type_\alpha,
\]
where the second inequality is implied by \eqref{R3L} evaluated at $\type_\alpha$.

The proposed mechanism induces $x^*(\type)=q^c(\type)$ everywhere. It
attains equality in the lower bound because
$\capa^*(\type)=x^*(\type)$ below $\type_\alpha$. It attains equality in
the upper bound because its subsidy is zero above $\type_\alpha$.
Consequently, it attains the right-hand side of
\eqref{eq:r3-global-bound} at its pointwise maximum. It therefore weakly
dominates every globally incentive-compatible no-tax mechanism.

Finally, strict decrease of $q^c$ implies
\[
q^c(\type)<q^c(\type_\alpha)=K_\alpha,
\text{ for every }\type>\type_\alpha.
\]
Because $F(\type_\alpha)=\alpha$, this strict-underutilization interval has
the measure $1-\alpha$.
\hfill \qedsymbol

\section{Proof of \autoref{prop_entry_fee} and \autoref{cor_entry_fee_underuse}}
We use the definitions and reduced program in \eqref{F}.
At a differentiable interior optimum at which these constraints are locally
slack, the entry cutoff satisfies
\begin{equation}
f(\type_E)
 \left[S(q^c(K, \type_E))+v_K(\type_E)\right]
=(1-\alpha)F(\type_E)q^c(K, \type_E).
\label{R4-E}\tag{E}
\end{equation}
The left-hand side is the density-weighted surplus from admitting the marginal type. The right-hand side is the net welfare cost of lowering the common fee for all inframarginal entrants.  Thus the fee screens entry but,
because it is the same for every entrant, does not screen capacity conditional
on entry.

\begin{proof}[Proof of \autoref{prop_entry_fee}]
 Given a menu $\menu$ containing zero, let
\[
    v_{\menu}(\type):=\max_{\capa\in\menu}\pi(\capa,\type)
\]
denote gross operating profit before the fee. For every fixed capacity, the envelope theorem gives
\[
    \frac{\partial\pi(\capa,\type)}{\partial\type}
    =-q^c(\capa,\type)\leq0.
\]
The pointwise maximum $v_{\menu}(\type)$ is therefore weakly decreasing in $\type$.
Apart from payoff-irrelevant zero-output ties, the set of types satisfying
$v_{\menu}(\type)\geq\fee$ is consequently a lower interval
$[\underline{\type},\type_E]$.  An optimal policy can be selected so that
\begin{equation}
    \fee=v_{\menu}(\type_E).
    \label{R4-PC}\tag{PC-F}
\end{equation}
For an interior cutoff, \eqref{R4-PC} is marginal-type indifference.  If all
types enter, the fee can be raised to $v_{\menu}(\overline{\type})$
without changing entry. Conditional on entry, the fee subtracts the same constant from every
capacity choice:
\[
[\pi(\capa,\type)-\fee]
-[\pi(\widehat{\capa},\type)-\fee]
=\pi(\capa,\type)-\pi(\widehat{\capa},\type).
\]
It therefore does not affect capacity or output conditional on entry.  The
regulator's payoff from an entering type is
\[
S(q^c(\capa_\menu(\type),\type))
+\alpha[v_{\menu}(\type)-\fee]+\fee
=S(q^c(\capa_\menu(\type),\type))+\alpha v_{\menu}(\type)
+(1-\alpha)\fee.
\]
Using \eqref{R4-PC}, expected welfare is
\begin{equation}
\int_{\underline{\type}}^{\type_E}
[S(q^c(\capa_\menu(\type),\type))+\alpha v_{\menu}(\type)]\dint F(\type)
+(1-\alpha)F(\type_E)v_{\menu}(\type_E).
\label{R4-menu-objective}
\end{equation}

Fix the entry cutoff $\type_E$ and the marginal entrant's gross profit
$r=v_{\menu}(\type_E)=\fee$.  The last term in
\eqref{R4-menu-objective} is constant in the allocation problem.  On the
truncated type interval $[\underline{\type},\type_E]$, the remaining
problem is exactly the fixed-terminal problem used in the proof of
\autoref{thm_capa_floor}, with terminal gross profit $r$.  Condition
\eqref{FL} on the full support implies its truncated-support counterpart.
Indeed, under \eqref{FL}, if $f'\geq0$ the
relevant density expression is nonnegative directly; if $f'<0$, then
$q^{LF}(\type_E)\geq q^{LF}(\overline{\type})$ and $S''\leq0$
make the truncated-support inequality weakly easier to satisfy.
The fixed-terminal argument therefore replaces any feasible active-type
allocation by its upper-root capacity floor while preserving $r$ and weakly
raising \eqref{R4-menu-objective}.

The replacement floor has gross value decreasing in type.  Its value at
$\type_E$ is $r=\fee$, so types above $\type_E$ continue to reject.
For the upper-root floor, $K\geq q^{LF}(\type_E)$, and the cutoff type
chooses $K$.  Hence
\[
    \fee=v_K(\type_E)=\pi(K,\type_E)\geq0,
\]
and substitution gives \eqref{F}.

For completeness, the feasible set
\[
\mathcal D:=\left\{(\type_E,K):
\underline{\type}\leq\type_E\leq\overline{\type},\ 
K\geq q^{LF}(\type_E),\ 
\pi(K,\type_E)\geq0\right\}
\]
is closed and bounded.  Boundedness follows from unbounded capacity cost and
the bounded maximum operating revenue.  The reduced objective is continuous
on $\mathcal D$, so the extreme value theorem implies a maximizer exists.  A zero-fee laissez-faire policy
gives strictly positive welfare, whereas $\type_E=\underline{\type}$
gives zero expected welfare; thus a maximizer can be selected with
$\type_E^*>\underline{\type}$.

Finally, holding the floor fixed and differentiating the cutoff objective,
the envelope theorem gives $v_K'(\type_E)=-q^c(K, \type_E)$.  Therefore
\[
\frac{\partial\mathcal J}{\partial\type_E}
=f(\type_E)[S(q^c(K, \type_E))+v_K(\type_E)]
-(1-\alpha)F(\type_E)q^c(K, \type_E),
\]
which yields \eqref{R4-E} at an interior optimum with locally slack
constraints.
\end{proof}

\begin{proof}[Proof of \autoref{cor_entry_fee_underuse}]
Fix the optimal cutoff $\type_E^*$, suppress the superscript, and write
\[
    a(K):=[KP(K)]'-\capacost'(K).
\]
For
\[
    q^{LF}(\type_E)<K\leq q_E,
\]
all participating types affected by the floor fully use it.  Moreover,
$a(K)\geq a(q_E)=a_E\geq\underline{\type}$.  Differentiating the objective in 
\eqref{F} while holding $\type_E$ fixed gives
\begin{align}
\mathcal J_K(\type_E,K)
={}&
\int_{a(K)}^{\type_E}
 [S'(K)+\alpha(a(K)-\type)]f(\type)\dint\type \notag +(1-\alpha)F(\type_E)[a(K)-\type_E] \notag\\
={}&
\int_{a(K)}^{\type_E}
 \{S'(K)f(\type)+\alpha F(\type)-F(\type_E)\}
 \dint\type .
\label{R4-D}
\end{align}
The second equality follows by integration by parts.  Relative to the proof
of \autoref{prop_underuse}, the term involving $F(\type_E)$ is the
capitalization effect: increasing the floor lowers the marginal entrant's
gross profit and therefore lowers the common fee paid by every entrant.

Define
\[
    \Phi_K(\type):=S'(K)f(\type)+\alpha F(\type).
\]
Condition \eqref{FL} implies that $\Phi_K$ is nondecreasing in $\type$
for every relevant $K$.  It also implies that
$K\mapsto\Phi_K(a(K))$ is nonincreasing because, wherever derivatives
exist,
\[
\frac{\dint}{\dint K}\Phi_K(a(K))
=S''(K)f(a(K))
+a'(K)[S'(K)f'(a(K))+\alpha f(a(K))]
\leq0.
\]
Here $S''\leq0$, $a'<0$, and the expression in brackets is nonnegative
under \eqref{FL}.  As in the preceding proof, this last implication is
immediate for $f'\geq0$ and follows from the full-support
inequality when $f'<0$.

Because $S'(q)=-qP'(q)$, the second inequality contained in
\eqref{U-E} is equivalent to
\[
    \Phi_{q_E}(a_E)>F(\type_E).
\]
Consequently, for every
$K\in(q^{LF}(\type_E),q_E]$ and every
$\type\in[a(K),\type_E]$,
\[
\Phi_K(\type)\geq\Phi_K(a(K))
\geq\Phi_{q_E}(a_E)>F(\type_E).
\]
Equation \eqref{R4-D} is therefore strictly positive throughout this
interval.

The first inequality contained in \eqref{U-E} gives
\[
\pi(q_E,\type_E)
=q_ES'(q_E)-\capacost(q_E)>0.
\]
Thus a small increase in the floor above $q_E$ preserves participation of
the cutoff type.  The right derivative at $K=q_E$ equals the limit in
\eqref{R4-D} and is strictly positive.  It follows that an optimal floor
conditional on $\type_E$, and hence $K^*$, must satisfy $K^*>q_E$.
An optimal floor satisfies \(K^*\leq q^c(\underline\type)<q_{\max}\). Indeed, when \(K>q^c(\underline\type)\), every entrant produces its continuation output, and holding the cutoff fixed gives \(\mathcal J_K=-F(\type_E)\capacost'(K)<0\). 
Strict concavity of revenue then gives
\[
    [K^*P(K^*)]'<[q_EP(q_E)]'=\type_E.
\]
The positive-length interval
\[
\left(\max\{\underline{\type},[K^*P(K^*)]'\},\type_E\right]
\]
therefore consists of participating types that install $K^*$ and produce
$q^c(\type)<K^*$.

For the final claim, the viability part of \eqref{U-E}, strict monotonicity
of capacity cost, and unbounded capacity cost imply that the cutoff type has
a unique upper zero-profit capacity $\overline K(\type_E)>q_E$.  At that
capacity its output is $q_E$, and
\[
    \capacost(\overline K(\type_E))=q_ES'(q_E).
\]
Condition \eqref{LF-E} implies
\[
\capacost(\overline K(\type_E))
<\capacost(q^{LF}(\underline{\type})).
\]
Hence
\[
    K^*\leq\overline K(\type_E)<q^{LF}(\underline{\type}).
\]
Define
\[
\type_{FL}:=[K^*P(K^*)]'-\capacost'(K^*),\;
\type_U:=[K^*P(K^*)]'.
\]
Because $K^*<q^{LF}(\underline{\type})$,
$\type_{FL}>\underline{\type}$.  Moreover,
\[
\type_U-\type_{FL}=\capacost'(K^*)>0,\;
\type_E-\type_U>0,
\]
where the second inequality follows from $K^*>q_E$.  These inequalities
establish the three positive-length regions among entrants.  If
$\type_E<\overline{\type}$, the non-entry interval supplies the stated
fourth region.
\end{proof}

\section{Omitted Proofs of Other Results}
\subsection{Details about $q^c$, $q^{LF}$, and $q^{FB}$}\label{apx:q-m-q-lf}
In this section, we formally show that $q^c$ and $q^{LF}$ are unique for each $\type$, characterized by their first-order conditions, and are strictly decreasing in $\type$, respectively.

First, consider the unrestricted continuation output problem: \[
q^c(\type) = \arg\max_{q \in [0, \qmax]}\{qP(q) - \type q\}.
\] By \autoref{asmp:maintained}, the objective is strictly concave in $q$. Then the first-order condition $(qP(q))' = \type$ is sufficient for the optimal output $q^c$ if $q^c \in [0, \qmax]$. By \autoref{asmp:maintained}, $(qP(q))'$ is strictly decreasing. Moreover, $\lim_{q \uparrow \qmax} (qP(q))' < \underline{\type}$ and $\lim_{q \downarrow 0} (qP(q))' > \overline{\type} + \capacost'(0) \geq \overline{\type}$.
Therefore, the solution $q^c \in (0, \qmax)$ for each $\type \in \typespace$. The standard implicit function theorem shows that $q^c(\type)$ strictly decreases in $\type$.

Next, consider the laissez-faire problem:\[
\max_{q \in [0, \capa], \capa \in [0, \qmax]}\{qP(q) - \type q - \capacost(\capa)\}.
\]
We claim that $\capa = q$ in the optimum. To see this, suppose that $\capa>q$ in the optimum. Then, the firm can make strictly more profit by lowering $\capa$ by some small $\epsilon > 0$ so that the capacity constraint is still slack, contradicting that $\capa > q$ is true in an optimal solution. 

With $\capa = q$, the problem becomes 
\[
\max_{q \in [0, \qmax]}\{qP(q) - \type q - \capacost(q)\}.
\]
The problem is the same as the one of unrestricted continuation output $q^c$, except that $\capacost(q)$ is not sunk. By \autoref{asmp:maintained}, $\capacost(\cdot)$ is convex, implying the objective is strictly concave. The sufficient first-order condition for an interior solution is $(qP(q))' - \capacost'(q) = \type$. 
By \autoref{asmp:maintained}, we have $\lim_{q \uparrow \qmax} (qP(q))' - \capacost'(q) < \underline{\type}$ and $\lim_{q \downarrow 0} (qP(q))' - \capacost'(q)  > \overline{\type}$. Therefore, $q^{LF} \in (0, \qmax)$ for each $\type \in \typespace$. The standard implicit function theorem shows that $q^{LF}(\type)$ strictly decreases in $\type$.

Finally, we observe that $q^c(\type) > q^{LF}(\type)$ for all $\type \in \typespace$. To see this, note that 
\[
[(qP(q))' - \capacost'(q)]_{q = q^{LF}(\type)} = \type= (qP(q))'|_{q = q^c(\type)}
\]
implies $(qP(q))'|_{q = q^{LF}(\type)} > (qP(q))'|_{q = q^c(\type)}$
for all $\type \in \typespace$.
Strict concavity of $qP(q)$ implies $q^c(\type)> q^{LF}(\type)$.

\medskip

In addition to the continuation output
$q^c(\type)$ and the laissez-faire output $q^{LF}(\type)$, define the
\emph{complete-information, tight-capacity benchmark}
\[
q^{FB}(\type)
:=
\arg\max_{q\in[0,\qmax]}W(q,q;\type)
=
\arg\max_{q\in[0,\qmax]}
\left\{
S(q)+\alpha
\left[qP(q)-\type q-\capacost(q)\right]
\right\},
\]
where $S(q):=V(q)-qP(q)$ is the consumer surplus.
This is the regulator's preferred output for a known type when capacity is
installed tightly, $\capa=q$.  It is a complete-information benchmark rather
than an allocation that is necessarily feasible under private information and
voluntary participation.

Under the demand-curvature restriction
$P'(q)+qP''(q)\geq0$ used in \eqref{FL}, consumer surplus is concave because
\[
S''(q)=-[P'(q)+qP''(q)]\leq0.
\]
Hence the objective defining $q^{FB}(\type)$ is concave and is strictly
concave whenever $\alpha>0$.  When $\alpha=0$, consumer surplus is strictly
increasing and $q^{FB}(\type)=\qmax$.  If the solution is interior, it
satisfies
\[
S'(q^{FB}(\type))
+
\alpha
\left[
\frac{d[qP(q)]}{dq}\bigg|_{q = q^{FB}(\type)}
-\type-\capacost'(q^{FB}(\type))
\right]
=0.
\]
At the laissez-faire output, the expression inside brackets is zero, whereas
\[
S'(q^{LF}(\type))
=
-q^{LF}(\type)P'(q^{LF}(\type))
>0.
\]
Moreover, concavity implies that the regulator's objective is increasing
throughout the interval up to $q^{LF}(\type)$.  It follows that
\[
q^{LF}(\type)<q^{FB}(\type).
\]
As established above, $q^{LF}(\type)<q^c(\type)$ as well.  Thus both the
regulator and the firm, after capacity cost has become sunk, prefer more output
than the firm installs under laissez-faire.

There is, however, no unconditional ordering between $q^{FB}(\type)$ and
$q^c(\type)$.  Evaluating the regulator's marginal objective at the
continuation output and using
$\frac{d[qP(q)]}{dq}|_{q = q^c(\type)}=\type$ gives
\[
\left.
\frac{\partial W(q,q;\type)}{\partial q}
\right|_{q=q^c(\type)}
=
S'(q^c(\type))
-\alpha\capacost'(q^c(\type)).
\]
Consequently,
\[
q^{FB}(\type)\gtreqless q^c(\type)
\quad\Longleftrightarrow\quad
-q^c(\type)P'(q^c(\type))
\gtreqless
\alpha\capacost'(q^c(\type)).
\]
The complete-information benchmark exceeds continuation monopoly output when
the regulator's marginal consumer-surplus gain dominates the weighted marginal
capacity cost.  The reverse ordering obtains when capacity is sufficiently
costly. 

This distinction clarifies the limitation of capacity control relative to
direct output regulation.  If output were contractible and the type were known,
the regulator could prescribe $q^{FB}(\type)$ and install exactly that much
capacity, whenever participation permits.  A capacity requirement instead
operates only through the firm's continuation problem.  Given installed
capacity $\capa$, output is
\[
q^c(\capa,\type)
=
\min\{\capa,q^c(\type)\}.
\]
A capacity requirement below $q^c(\type)$ therefore binds, and increasing it
raises output.  Once capacity exceeds $q^c(\type)$, however, the firm's
continuation choice is unconstrained: additional capacity raises cost but not
output.  Hence capacity control can implement the complete-information tight
benchmark for a known type when
$q^{FB}(\type)\leq q^c(\type)$, subject to participation, but it cannot induce
output above $q^c(\type)$.

Under private information, the regulator faces the additional problem that one
menu must govern all types.  The same capacity requirement may bind the output
of some types while becoming slack for others.  The optimal policy must
therefore balance the output gain generated for types whose capacity remains
binding against the capacity cost imposed on types for whom the same
requirement becomes slack.

\subsection{Proof of \autoref{lem:monotone-welfare}}

By \autoref{lem:opt_ctrl}, capacity is weakly decreasing in type in every
feasible no-transfer mechanism. Since $q^c(\type)$ also decreases,
$q^c(\capa(\type),\type)=\min\{\capa(\type),q^c(\type)\}$ is weakly
decreasing. The envelope identity gives
$\Pi(\type)=\Pi(\overline\type)+\int_\type^{\overline\type}
q^c(\capa(t),t)\dint t$, so truthful profit is weakly decreasing.
Finally, $S'(q)=-qP'(q)>0$ for positive output. Consumer surplus therefore
weakly decreases with type, as does $w=S(q)+\alpha\Pi$ because $\alpha\geq0$.
\hfill \qedsymbol
\subsection{Proof of \autoref{lem:full-delegation}}
Under no regulation ($\menu^{LF} = [0, \infty)$), the firm of type $\type$ chooses tight-capacity output $q^{LF}(\type) >0$. According to the argument in Appendix \ref{apx:q-m-q-lf}, $q^{LF}$ solves $a(q):=(qP(q))' - \capacost'(q) = \type$ and satisfies $q_0 := q^{LF}(\overline{\type}) >0$.

Denote the capacity-floor policy by $\menu^{\text{floor}}(K) := \{0\}\cup [K, \infty)$ with $K > 0$. With fixed parameters $P(\cdot)$ and $\capacost(\cdot)$, observe that $\menu^{LF}$ and $\menu^{\text{floor}}(q_0)$ are outcome equivalent because $q^{LF}(\type) \geq q_0$ for all $\type \in \typespace$.
 We show that there exists some $K^* > q_0$ such that the expected welfare under $\menu^{\text{floor}}(K^*)$ is strictly larger than that under $\menu^{\text{floor}}(q_0)$.  By \autoref{asmp:maintained}, $a(q)$ is strictly decreasing in $q$. Then, for $K^* > q_0$ and sufficiently close to $q_0$, $\underline{\type} < a(K^*) < \overline{\type}$. We first verify that this local perturbation induces neither exit nor underutilization ($q < \capa$). According to Appendix \ref{apx:q-m-q-lf}, $q_0 < q^c(\overline{\type})$. In addition, strict concavity of the profit function $\pi(\cdot, \type)$ and $\pi(0, \overline{\type}) = 0$ imply $\pi(q_0, \overline{\type}) > 0$. Consequently, if $K^*$ is in a sufficiently small right neighborhood of $q_0$,\[
K^* < q^c(\overline{\type}),\; \pi(K^*, \overline{\type})>0.
\]
Under this floor, types below $a(K^*)$ continue to choose their preferred output (and capacity) $q^{LF}(\type)$, whereas types in $[a(K^*), \overline{\type}]$ choose and fully utilize $K^*$. The top type participates strictly, and hence every lower type participates as well. 

Next, we show that welfare is weakly higher for every affected type and
strictly higher for every strictly affected type. Since
\(S'(q_0)=-q_0P'(q_0)>0\), continuity allows us to choose
\(\epsilon>0\), sufficiently small that the participation and full-utilization
conditions above continue to hold, and such that
\[
m_\epsilon
:=
\min_{z\in[q_0,q_0+\epsilon]} S'(z)
-\alpha\bigl[a(q_0)-a(q_0+\epsilon)\bigr]
>0.
\]
Fix \(K^*\in(q_0,q_0+\epsilon)\). For any affected type
\(\type\in[a(K^*),\overline{\type}]\), let
\(x:=q^{LF}(\type)\). Then
\[
a(x)=\type,\;
q_0\leq x\leq K^*.
\]
For every \(y\in[x,K^*]\), output and capacity are both equal to \(y\).
Therefore,
\[
\frac{d}{dy}W(y,y;\type)
=
S'(y)+\alpha[a(y)-\type]
=
S'(y)+\alpha[a(y)-a(x)].
\]
Because \(a\) is decreasing and
\(q_0\leq x\leq y\leq K^*<q_0+\epsilon\), we have
\[
a(y)-a(x)
\geq
a(q_0+\epsilon)-a(q_0).
\]
Consequently,
\[
\frac{d}{dy}W(y,y;\type)
\geq
\min_{z\in[q_0,q_0+\epsilon]}S'(z)
-\alpha\bigl[a(q_0)-a(q_0+\epsilon)\bigr]
=
m_\epsilon
>0.
\]
It follows that
\[
\begin{split}
W(K^*,K^*;\type)-W(x,x;\type)
&=
\int_x^{K^*}\frac{d}{dy}W(y,y;\type)\dint y  \\
&\geq
m_\epsilon(K^*-x).
\end{split}
\]
This difference is strictly positive for every
\(\type\in(a(K^*),\overline{\type}]\), because then \(x<K^*\);
it equals zero only for the boundary type \(\type=a(K^*)\).
All lower types are unaffected. Since
\(a(K^*)<\overline{\type}\) and \(f>0\), the strictly affected interval
has positive probability. Hence the capacity floor at \(K^*\) yields
strictly greater expected welfare than full delegation. 
\hfill \qedsymbol

\subsection{Comparative statics of the capacity floor}
\label{ext:comparative-statics}
 Let $K^*$ be the
smallest optimal full-participation floor. We show that $K^*$ weakly falls
with $\alpha$, with $\lambda$ in $\capacost_\lambda=\lambda\capacost$,
and with higher marginal costs in hazard-rate order. It weakly rises under
parallel upward inverse-demand shifts and proportional market-size
expansions.

\begin{proof}
\autoref{thm_capa_floor} allows us to restrict attention to fully participating floors with
\[
q^{LF}(\bar\theta)\leq K\leq q^c(\underline\theta),
\quad \pi(K,\bar\theta)\geq0.
\]
Write \(R(q):=qP(q)\), and extend \(F\) and \(f\) by zero below \(\underline\theta\).
 For $K>q^{LF}(\bar\theta)$,
differentiating expected floor welfare $J(K)$ and integrating by parts gives
\[
J'(K)=\int_{R'(K)-\capacost'(K)}^{\min\{\bar\theta,R'(K)\}}
\bigl\{S'(K)f(\theta)-\alpha[1-F(\theta)]\bigr\}\,d\theta.
\]
Regime-boundary terms cancel because output and profit are continuous.
Condition \eqref{FL} makes the integrand nondecreasing in $\theta$: its derivative
is $S'(K)f'+\alpha f\geq0$ on the support, and it jumps upward at
$\underline\theta$. It is nonincreasing in $K$ because $S''\leq0$.
Both integration limits also decrease weakly with $K$. Thus the integrand's
average over this interval is nonincreasing in $K$. This average is $J'(K)$
divided by the positive interval length, so it has the same sign as $J'(K)$.
Use right limits at $q^{LF}(\bar\theta)$, where both limits equal
$\bar\theta$. Welfare therefore increases and then decreases, allowing an
optimal plateau. Ordering this sign crossing and the feasible endpoints
orders its smallest maximizer.

Increasing $\alpha$ lowers the integrand pointwise while leaving feasibility
and both integration limits unchanged, so $K^*$ falls weakly. Increasing
$\lambda$ lowers the lower limit $R'(K)-\lambda\capacost'(K)$, leaving
the upper limit and integrand unchanged. Extending the averaging interval
leftward lowers the average. Both feasible endpoints also fall weakly:
laissez-faire capacity decreases, and profit at every positive capacity falls,
reducing the participation bound. This proves the cost comparison.
For these two changes, slack equals
$[K^*-q^c(\theta)]_+$, so it also weakly decreases for every type.

For the prior comparison, on the common feasible interval above
$q^{LF}(\bar\theta)$,
\[
\begin{aligned}
\frac{J_i'(K)}{1-F_i(R'(K)-\capacost'(K))}
={}&S'(K)\left[1-
\frac{1-F_i(\min\{\bar\theta,R'(K)\})}
     {1-F_i(R'(K)-\capacost'(K))}\right]\\
&-\alpha\int_{R'(K)-\capacost'(K)}^{\min\{\bar\theta,R'(K)\}}
\frac{1-F_i(\theta)}{1-F_i(R'(K)-\capacost'(K))}\,d\theta.
\end{aligned}
\]
If $f_2/(1-F_2)\leq f_1/(1-F_1)$ on $[\underline\theta,\bar\theta)$,
every conditional-survival ratio on the right is weakly larger under $F_2$.
The normalized derivative is therefore lower, proving the claim.
If $f_2/f_1$ is nondecreasing, its upper-tail average
$(1-F_2)/(1-F_1)$ is at least its current value, which implies the
hazard-rate inequality.

Finally, suppose $R_2'\geq R_1'$ and $S_2'\geq S_1'$ from zero throughout
the old quantity domain, with a common or expanded new domain. Both integration
limits rise, and the integrand rises pointwise, so its average weakly rises
at every common feasible floor. Laissez-faire and continuation capacities
rise. Integrating marginal-revenue dominance from $R_i(0)=0$ gives revenue
dominance, raising the participation bound as well. Thus $K^*$ rises weakly.
For $P_A=P_0+A$, marginal revenue rises by $A$ and $S'$ is unchanged.
For $P_z(q)=P_0(q/z)$, we have $R_z'(q)=R_0'(q/z)$ and $ S_z'(q)=S_0'(q/z)$.
Both increase weakly with $z$, since $R_0''<0$ and $S_0''\leq0$.
\end{proof}

\subsection{Common entry fee charged before the firm learns its type}\label{apx:pretype_fee}

Before
\(\type\) is realized, the regulator announces a pair \((\menu,\fee)\), and
the firm accepts or rejects.  If it accepts, the fee \(\fee\geq 0\) is
nonrefundable and remains sunk even if the firm exits after observing its
type.  The firm is risk neutral, retains its zero outside option, and accepts
when indifferent.  Fee revenue is valued one-for-one by the regulator.

Denote the firm's gross post-entry profit as \(\Pi(\type)\).  For
a feasible post-entry mechanism \((\capa,\Pi)\), ex ante participation
requires
\begin{equation}
    \fee\leq \intfullsupport \Pi(\type)\dint F(\type).
    \label{fee-PC}\tag{PC}
\end{equation}
Denote the consumer surplus as $S(q) := V(q) - qP(q)$.  The regulator's ex ante payoff is
\begin{align}
\mathcal W^{\fee}(\capa,\Pi,\fee)
&:=\intfullsupport
\left[
S(q^c(\capa(\type),\type))
+\alpha\bigl(\Pi(\type)-\fee\bigr)
\right]\dint F(\type)+\fee \notag\\
&=\intfullsupport
\left[
S(q^c(\capa(\type),\type))
+\alpha\Pi(\type)
\right]\dint F(\type)
+(1-\alpha)\fee .
\label{fee-objective}
\end{align}

\begin{proposition}
\label{prop_entry_fee_unknown}
An upfront fee leaves the set of feasible post-award mechanisms
unchanged.  If \(\alpha<1\), the ex ante participation constraint
\eqref{fee-PC} binds at every optimum:
\begin{equation}
    \fee^*=\intfullsupport\Pi^*(\type)\dint F(\type).
    \label{optimal-fee}
\end{equation}

\end{proposition}

\begin{corollary}[Capacity floor and underutilization]
\label{cor:upfront-fee-floor}

If \(\alpha=1\), a binding fee can be selected without changing the
regulator's payoff.  Consequently, for every \(\alpha\in[0,1]\), the joint
allocation problem with an upfront award fee is exactly the baseline problem
\eqref{direct} evaluated at welfare weight one.

Suppose condition \eqref{FL} holds.  Then there exists an optimal award
consisting of a capacity floor \(\{0\}\cup[K^*, \infty)\) and an upfront fee
\begin{equation}
    \fee^*=\intfullsupport\Pi_{K^*}(\type)\dint F(\type),
    \qquad
    K^*\in
    \left(q^{LF}(\overline{\type}),q^c(\underline{\type})\right],
    \label{optimal-floor-fee}
\end{equation}
where \(\Pi_K(\type)\) denotes gross operating profit induced by the floor
\(\{0\}\cup[K, \infty)\) under the maintained response convention.

Let \(\overline q:=q^c(\overline{\type})\).  If, in addition,
\begin{equation}
    -\overline qP'(\overline q)>\capacost'(\overline q),
    \label{U-one}\tag{U'}
\end{equation}
then an optimal floor can be selected with \(K^*>\overline q\), and a
positive measure of types installs \(K^*\) but produces strictly less than
\(K^*\).  If condition \eqref{LF} also holds, the laissez-faire,
fully used-floor, and underutilized-floor regions all have positive length.

\end{corollary}
\begin{proof}[Proof of \autoref{prop_entry_fee_unknown}]
Fix a feasible post-award mechanism \((\capa,\Pi)\). After paying
\(\fee\), a type \(\type\) that reports \(\widehat\type\) obtains
\[
\pi(\capa(\widehat\type),\type)-\fee,
\]
whereas post-award exit yields \(-\fee\). Hence the fee cancels from
every reporting comparison, every capacity and output comparison, and
every comparison between operation and exit. In particular,
\[
\Pi(\type)-\fee
\geq
\pi(\capa(\widehat\type),\type)-\fee
\]
is equivalent to \eqref{IC}, and
\[
\Pi(\type)-\fee\geq-\fee
\]
is equivalent to \eqref{IR}. Because \(\Pi\) denotes gross post-award
profit, the no-transfer identity \eqref{N} is also unchanged.

Conversely, every mechanism satisfying \eqref{IC}, \eqref{IR}, and
\eqref{N} can be paired with \(\fee=0\). Thus the projection of the
joint feasible set onto post-award mechanisms is exactly the baseline
feasible set.

The firm's expected net payoff from accepting the award is
\[
\intfullsupport\Pi(\type)\dint F(\type)-\fee.
\]
Therefore ex ante participation is equivalent to \eqref{fee-PC}. If
\(\alpha<1\) and \eqref{fee-PC} were slack, the regulator could raise
\(\fee\) slightly without changing any continuation choice or violating
participation. By \eqref{fee-objective}, this would increase the
regulator's payoff by \(1-\alpha>0\) times the increase in the fee.
Hence \eqref{fee-PC} binds at every optimum:
\[
\fee^*
=
\intfullsupport\Pi^*(\type)\dint F(\type).
\]
When \(\alpha=1\), the coefficient on \(\fee\) is zero. Under the
maintained acceptance-at-indifference convention, a binding fee can
therefore be selected without changing either the allocation or the
regulator's payoff.

Finally, substituting the binding fee into \eqref{fee-objective} gives
\[
\intfullsupport
\left[
S(q^c(\capa(\type),\type))+\Pi(\type)
\right]\dint F(\type).
\]
This is precisely the baseline problem \eqref{direct} evaluated at
welfare weight one.
\end{proof}

\begin{proof}[Proof of \autoref{cor:upfront-fee-floor}]
By \autoref{prop_entry_fee_unknown}, the allocation problem is the
baseline problem evaluated at welfare weight one. Let
\(\hat q=q^{LF}(\overline{\type})\). Since $S'(\hat q)=-\hat qP'(\hat q)>0$
and \(\alpha\leq1\), the density inequality in \eqref{FL} implies,
almost everywhere,
\[
-\frac{f'(\type)}{f(\type)}
\leq
\frac{\alpha}{S'(\hat q)}
\leq
\frac{1}{S'(\hat q)}.
\]
The curvature component of \eqref{FL} is independent of the welfare
weight. Thus \eqref{FL} continues to hold when the welfare weight is
raised from \(\alpha\) to one. Applying
\autoref{thm_capa_floor} at welfare weight one yields an optimal floor $\{0\}\cup[K^*, \infty)$ such that 
$K^*\in
(q^{LF}(\overline{\type}),
q^c(\underline{\type})]$.
By \autoref{prop_entry_fee_unknown}, its associated optimal fee is
\[
\fee^*
=
\intfullsupport\Pi_{K^*}(\type)\dint F(\type).
\]

Now let \(\overline q=q^c(\overline{\type})\). Strict convexity of
\(\capacost\), together with \(\capacost(0)=0\), gives
\[
\frac{\capacost(\overline q)}{\overline q}
=
\frac{1}{\overline q}
\int_0^{\overline q}\capacost'(x)\dint x
<
\capacost'(\overline q).
\]
Consequently, \eqref{U-one} implies
\[
-\overline qP'(\overline q)
>
\max\left\{
\frac{\capacost(\overline q)}{\overline q},
\capacost'(\overline q)
\right\}.
\]
This is condition \eqref{U} evaluated at welfare weight one.
Applying \autoref{prop_underuse} at welfare weight one therefore yields
an optimal floor with \(K^*>\overline q\) and a positive measure of
underutilizing types.

Finally, \eqref{LF} does not depend on the welfare weight. When it also
holds, the three-region conclusion follows directly from
\autoref{prop_underuse}.
\end{proof}

\end{document}